\documentclass[11pt, letterpaper]{amsart}
\usepackage{color}
\usepackage{enumerate}
\usepackage{amssymb}
\usepackage{epsfig,psfrag,latexsym}
\usepackage{xcolor}
\usepackage{natbib}
\usepackage{ulem}
\makeatother
\newtheorem{theorem}{Theorem}[section]
\newtheorem{lemma}[theorem]{Lemma}
\newtheorem{proposition}[theorem]{Proposition}

\newtheorem{assumption}[theorem]{Assumption}

\usepackage{color}
\usepackage{indentfirst}

\theoremstyle{remark}
\newtheorem{remark}[theorem]{Remark}
\theoremstyle{definition}

\newcommand{\nada}[1]{}

\newcommand{\abs}[1]{\left| #1 \right|} % absolute value
\newcommand{\PP}{{\mathbb P}}
\newcommand{\RR}{{\mathbb R}}
\newcommand{\QQ}{{\mathbb Q}}

\newcommand{\calv}{\mathcal{V}}

\newcommand{\MM}{{\mathcal M}}

\newcommand{\YY}{{\mathcal Y}}
\newcommand{\reals}{\mathbb{R}}
\newcommand{\prob}{\mathbb{P}}

\newcommand{\A}{{\mathcal A}}

\newcommand{\E}{\mathcal{E}}
\newcommand{\F}{\mathcal{F}}

\newcommand{\rdfn}{\, =: \,}

\newcommand{\cbra}[1]{\left\{#1\right\}}
\newcommand{\dfn}{:=}
\newcommand{\wh}[1]{\widehat{#1}}

\newcommand{\wt}[1]{\widetilde{#1}}

\newcommand{\ceinto}[1]{\textrm{CE}^{#1}_{0}}

\newcommand{\cexao}[1]{\textrm{CE}^{#1}_{0-}}

\newcommand{\ceintio}[1]{\textrm{CE}^{#1}_{0,\iota}}

\newcommand{\cexaio}[1]{\textrm{CE}^{#1}_{0-,\iota}}
\newcommand{\recip}[1]{\frac{1}{#1}}

\newcommand{\expv}[3]{\mathbb{E}^{#1}_{#2}\left[#3\right]}
\newcommand{\condexpv}[4]{\mathbb{E}^{#1}_{#2}\left[#3\big| #4\right]}
\newcommand{\expvs}[1]{\mathbb{E}\left[#1\right]}
\newcommand{\condexpvs}[2]{\mathbb{E}\left[#1\big| #2\right]}

\newcommand{\condprobs}[2]{\mathbb{P}\left[#1\big| #2\right]}
\newcommand{\qcondprobs}[2]{\mathbb{Q}_0\left[#1\big| #2\right]}

\newcommand{\bra}[1]{\left[#1\right]}
\newcommand{\ol}[1]{\overline{#1}}

\newcommand{\such}{\ | \ }

\newcommand{\rto}{\lambda}
\newcommand{\rtoi}{\lambda_I}
\newcommand{\rtou}{\lambda_U}
\newcommand{\cons}{\beta}
\newcommand{\consa}{\kappa}

\newcommand{\yval}{\YY}
\newcommand{\yvals}{y}
\newcommand{\consb}{\lambda}
\newcommand{\pix}{P_{X|G}}

\newcommand{\pnsn}{p_{\QQ_0}}
\newcommand{\whpi}{\widehat{\Pi}}
\newcommand{\whyval}{\widehat{\yval}}

\begin{document}

\title[Strategic Informed Trading and the Value of Private Information]{Strategic Informed Trading and the Value of Private Information}

\author{Michail Anthropelos}
\address{Department of Banking and Financial Management\\
University of Piraeus\\
Piraeus, Greece}
\email{anthropel@unipi.gr}

\author{Scott Robertson}
\address{Questrom School of Business\\
Boston University\\
Boston, MA 02215}
\email{scottrob@bu.edu}

\begin{abstract}
We consider a market of risky financial assets whose participants are an informed trader, a representative uninformed trader, and noisy liquidity providers. We prove the existence of a market-clearing equilibrium when the insider internalizes her power to impact prices, but the uninformed trader takes prices as given. Compared to the associated competitive economy, in equilibrium the insider strategically reveals a noisier  signal, and prices are less reactive to publicly available information. Additionally, and in direct contrast to the related literature, in equilibrium the insider's  indirect utility monotonically increases in the signal precision. Therefore, the insider is motivated not only to obtain, but also to refine, her signal. Lastly, we show that compared to the competitive economy, the insider's internalization of price impact is utility improving for the uninformed trader, but somewhat  surprisingly may be utility decreasing for the insider herself. This utility reduction occurs provided the insider is sufficiently risk averse compared to the uninformed trader, and provided the signal is of sufficiently low quality.
\end{abstract}

\date{\today}
\maketitle

\section*{Introduction}\label{sec:intro}

It is well-documented that large financial institutions possess the power to affect markets (for example,  see \cite{KoiYog19} and \cite{RosYoo23}). Compared to other traders, large investors' orders impact both transaction prices and volumes, and these investors are aware of this impact (\cite{RosWer15}). Additionally, financial institutions invest considerable capital to acquire information regarding traded asset payoffs (\cite{KacPag19}). In sum, it is natural to assume large investors are both aware of their impact on prices and in possession of private information. 

However, it is not a secret that large investors are informed traders (\cite{subrahmanyam1991risk}). Indeed, other market participants, even if they lack private information themselves, both know and account for the large investors' (``insiders'') superior information. This means that in equilibrium, one expects a partial transmission of private information.  Indeed, this was shown concretely in both competitive economy and price impact models (starting with  \cite{grossman1980impossibility} and \cite{kyle1985continuous} respectively), as therein the insider's private signal about the risky assets' terminal payoff is partially revealed to all market participants through equilibrium prices, a mechanism that creates a market (or public) signal. Therefore, it is also reasonable to assume uninformed traders know the insider's private information  will be, at least partially, revealed through market prices. 

With the above as motivation, we study how the insider's awareness of price impact affects equilibrium price formulation and information transmission, as well as the indirect utility of both the insider and uninformed trader. We work in a single period normal-CARA model, seeking a linear price-impact equilibrium where a risk averse insider trades a bundle of risky assets with both a mass of uninformed risk-averse traders and liquidity providers, or noise traders. We differentiate the insider in two ways: first (clearly) in that she possesses private information; and second in that she accounts for the impact her trading activity has on prices. We then use this model to predict to what extent the insider's signal is revealed to the market and how the uninformed traders correspondingly adjust their demands. We also compare equilibrium quantities of this model to those in the corresponding competitive economy model, to assess whether, and by how much, the informational content within equilibrium prices is reduced and what this implies in terms of traders' demands.

We are particularly focused on the traders' indirect utility, defined as the certainty equivalent from optimal trading. Our interest stems from a striking result in \cite{grossman1980impossibility} (similar results are obtained in both \cite{verrecchia1982} and the recent \cite{NEZAFAT2023105664}) which shows that insider indirect utility need not be increasing in the quality of her signal.  More precisely, if the signal is $G = X + Z_I$ where $X$ is the asset payoff and $Z_I$ is a noise term with precision $p_I$, then the map ``$p_I \to \textrm{Insider Indirect Utility}$'' takes one of the two forms shown in Figure \ref{F:figure1}.
\begin{figure}[h!]
\begin{center}
\includegraphics[height=4cm,width=7cm]{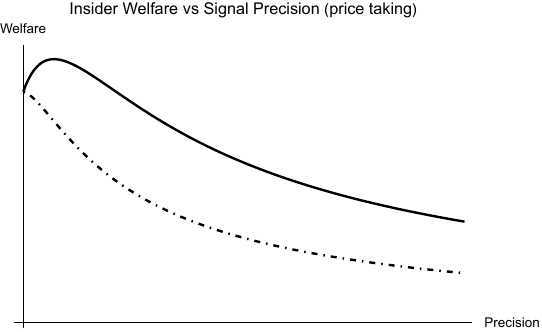}
\caption{ }\label{F:figure1}
\end{center}
\end{figure}
This is problematic, because either it is never advantageous to obtain the signal (dot-dashed line) or it is only beneficial to obtain the signal for precisions below a certain threshold (solid line).  As presumably there is a cost (in terms of money and/or effort) to produce the signal, the former case suggests price-taking equilibria with private signals are somewhat artificial, and the latter case requires the insider to estimate model parameters (in particular those related to other traders, which are difficult to estimate) to determine if it is worthwhile to refine her signal. 

Motivated by Figure \ref{F:figure1} we ask if our model produces a similar result, or if internalization of price impact ensures that the insider's indirect utility is increasing in her signal precision. If this is indeed the case, it would suggest the improvement is driven by the insider's differentiation in both access to better information and the internalization of price impact.

\subsection*{Methodology and Main Contributions}

We adjust the single period CARA-normal setting of \cite{grossman1980impossibility}\footnote{The competitive economy assumption in  \cite{grossman1980impossibility} is ubiquitous in the heterogeneous information asset pricing literature. It holds in the seminal papers of \cite{grossman1976efficiency}, \cite{grossman1980impossibility} and \cite{hellwig1980aggregation}, and with the exception of the literature strand started by \cite{kyle1985continuous, back1992insider}, \cite{rochet_vila} and \cite{subrahmanyam1991risk}, price taking has remained the dominant assumption.} by allowing the insider to internalize her price impact, while maintaining the presence of price-taking uninformed traders and liquidity providers\footnote{Assuming the uninformed traders are price takers is realistic, as they represent a mass of small risk-averse traders who rationally optimize their positions, but do not have the power to move prices.}. Following the related literature (e.g. \cite{kyle1985continuous, rochet_vila}), we study a linear impact equilibrium where the insider perceives the market price to be an affine function of the sum of her and the noise traders' demand\footnote{Linear price impact is common in the literature (see among others \cite{Kyl89}, \cite{Vay99}, \cite{Viv11}), and the affine structure of price impact is also seen (though not perceived a-priori by the insider) in the price-taking, or competitive, equilibrium (see \eqref{E:strat_map_simple} below).}, with the affine coefficients endogenously determined through market clearing.

In Theorem \ref{thm:pi_simple_n} we  establish existence of a linear price impact (PI) equilibrium, where the coefficients are governed by the unique positive root to the cubic equation in \eqref{E:cubic_alt_n}\footnote{This stands in contrast to \cite{subrahmanyam1991risk} where equilibria are governed by solutions to a quintic equation. The difference arises because, consistent with \cite{grossman1980impossibility} (see also \cite{rochet_vila}) we assume the insider, by viewing both her demand and the resultant price, is able to deduce the noise trader's demand.}.  To facilitate  comparison, we summarize the price-taking (PT) equilibrium results of \cite{grossman1980impossibility} in Proposition \ref{P:pt_equilibrium_simple_n}. Lastly, in both the PI and PT cases, we establish equilibria absent private information by showing it coincides with the zero precision limit (Proposition \ref{P:no_signal_simple_n}).

Comparing the PI and PT equilibria, we first show, similarly to \cite{KacNosSun23}, the public signal is fuzzier in the PI equilibrium (Proposition \ref{prop:impact_signal_noise}). This is reasonable, as the insider has a motive to hide her signal when submitting her demand, and hence her trading lowers the precision of the signal revealed by the prices. The uninformed traders duly recognize the insider reveals a muddied signal, and respond with a less elastic demand function. This in turn makes the PI equilibrium price less reactive to the public information than the PT price (consistent with the adverse-selection concerns of \cite{Kyl89} and \cite{LouRah23}). These results are robust across model parameter values, including agents' risk aversion.  We conclude that by assuming the insider is a price taker, one implicitly assumes the market receives a more precise signal, and prices are more reactive to public information, than  when one assumes the insider internalizes price impact.

We next consider the traders' indirect utility (henceforth referred to as ``utility'' for ease of exposition). Our main result shows that insider ex-ante utility is monotonically increasing in the private signal precision (Theorem \ref{T:pimono}).  Therefore, the relationships displayed in Figure \ref{F:figure1}, which hold when either (i) all agents are price takers (e.g. \cite{grossman1980impossibility}) or (ii) all agents internalize price impact (e.g. \cite{NEZAFAT2023105664}) do not occur in our model.  This means there is always an incentive for the insider to refine her signal, provided she is differentiated from the other agents in terms of both (i) price impact internalization and (ii) asymmetric information. 

It is of equal interest to understand the roles these two channels play in determining the insider utility in the first place. This question is motivated both by Figure \ref{F:figure1} which indicates that information asymmetry is utility reducing, and by the naive belief that internalizing price impact should be utility improving.  Thus, it is not immediately clear if on balance the PI equilibrium is utility improving over the PT equilibrium.  Here, we first confirm that absent private information, internalizing price impact improves insider's utility.  Remarkably, this holds for any realization of the noise trader demand (Proposition \ref{prop: nosigint_simple}). However, as surmised, the situation changes when one layers in the private signal.  While typically (i.e.~across the bulk of the parameter space) the utility benefit due to internalization of price impact outweighs the loss due to information asymmetry, and the PI model is utility improving over the PT model, this need not always be the case. Indeed, when the insider is sufficiently risk averse, the uninformed traders are sufficiently risk tolerant, and the variance of noise demand is sufficiently low\footnote{The exact condition is given in equation \eqref{E:less_CE_in_PI} below. Also, the variance of noise traders' demand is a proxy for the size of their trading, as discussed in \cite{KovViv14, NEZAFAT2023105664}.}, the negative effects due to information asymmetry dominate and the PI model is \textit{utility reducing} compared with the PT model. Therefore, there is no uniform utility ordering between the PI and PT equilibria, and in Section \ref{S:demand_compare} we give further reasoning why this may occur, focusing on the insider's demand function.

We have shown that a better signal is utility improving for the insider in the PI model, but may not be utility improving for the insider in the PT model. This result underscores an important, and clarifying, fact. As long as the insider internalizes her price impact, the PI equilibrium price cannot be driven to the corresponding PT equilibrium price. This is due to the uninformed traders' optimal demand. Despite being price takers themselves, when determining their optimal demand, they recognize that the insider internalizes her price impact. As  price impact reduces the public signal's precision, the uninformed traders demand is less elastic than in the PT equilibrium. The lower sensitivity to the public signal alters their demand function and hence, in the  PI equilibrium the insider trades against a different residual demand than in the PT equilibrium. Thus,  the PT equilibrium cannot be written as a special case of the PI one.

\subsection*{Connection with the related literature}

Our paper contributes to the on-going literature on price-impact equilibria under asymmetric information; on market participants' utility; and on the informativeness of equilibrium prices.

In the PI equilibrium, the insider does not act as price-taker. Usually, the price-taking assumption is made for tractability, as in its absence one must specify a price impact model, and depending upon the specification, it may be very difficult to establish equilibria.  In the aforementioned \cite{kyle1985continuous, back1992insider, rochet_vila}, the insider's demand is combined with exogenous noise traders' demand before being sent to a risk neutral market maker who prices in a competitive environment.  In \cite{subrahmanyam1991risk}, market makers are risk averse and quote prices to remain at utility indifference (as opposed the uninformed agent of \cite{grossman1980impossibility} who can be thought of as a market maker who quotes utility-optimal prices), but the insider does not know the noise trader demand before submitting her order (this is also the case in \cite{kyle1985continuous}). In fact, our analysis updates \cite{subrahmanyam1991risk} in two directions. First, by assuming the uninformed trader is a utility optimizer, and second, by allowing the insider to the identify the noise trader demand through the public equilibrium price (as in \cite{rochet_vila}).

Following the seminal work of \cite{Kyl89}, several models in the normal-CARA setting with price-impact and asymmetric information have been developed. For example, \cite{Vay01} and  \cite{RosWer15} study dynamic thin markets with and without market makers respectively; \cite{RosWer15a} focuses on the traders' interdependent preferences and correlated private signals; and \cite{MalRos17} and \cite{AnthKar24} consider decentralized exchanges and restricted participation settings. \cite{BerHeuMor21} studies a market of divisible goods where agents receive correlated signals, and their demand affects the revealed signal at equilibrium as in our model. In these works, as in \cite{Kyl89} and \cite{Viv11}, strategic agents submit demand schedules which leads to a Nash equilibrium. In contrast to our model, therein all (non-noise) agents are strategic, and private information appears in the investors' endowments. An overview of this literature is provided in \cite{RosYoo23}.

While large financial institutions invest to obtain private information, even when trading in markets that are not thin, theoretical studies such as ours that justify the positive relation between better information and higher gains from trading are scarce. In a competitive market, the fact that private information has positive value if it is internalized by an investor was pointed out in \cite{Hirsh71}, with similar positive effects shown in a competitive model with a continuum of traders in \cite{MorShi02}. In non-competitive market settings, information acquisition has been studied in \cite{Viv11}, \cite{RosWer12} and \cite{Viv14}. An extension of these papers (as well as \cite{verrecchia1982} which takes place in a competitive economy) to a two-stage model has been recently developed in \cite{NEZAFAT2023105664}, where in the first stage one group of agents selects the precision of their private signal before trading. However, unlike our model, both the insider and uninformed traders are assumed strategic, which essentially makes the market thin. Strategic uninformed traders, together with specific conditions on noise traders’ demand, lead to an equilibrium at which private information is utility-deteriorating. By contrast, we
show that the existence of zero-information equilibria are not possible if the uninformed traders are price takers, meaning that private information does have positive value for the insider.

\cite{KacNosSun23} considers strategic informed and price-taking uninformed traders, as we do, where the role of initial endowments is highlighted for different types of informed traders.  In \cite{GonKeQuiShen22} the insider is assumed risk neutral and the role of uninformed traders is played by a market maker with a quadratic objective. A linear equilibrium, similar to \cite{kyle1985continuous}, is derived where the price under-reacts to public signal, as in our case. Strategic agents and asymmetric information have been included in \cite{LouRah23}, where in a non-competitive market (in line with \cite{Kyl89} and \cite{RosWer15}) traders receive different ex-ante random values of a single asset (and potentially different private signals too). In contrast to our model (and to \cite{Kyl89}) there are no liquidity providers, which essentially implies that price informativeness does not change due to price impact. As in our model, there are conditions that lead to higher ex-ante expected utility for the uninformed traders than the insider. We reach a similar conclusion, but without assuming that uninformed traders act strategically.

Lastly, another channel that reduces insider's gains from trading due to private information is that of information sharing. For example, \cite{GolXioYan23} obtain this result in a model (based on \cite{Kyl89}) where informed traders share their private signals before trading (as in \cite{IndLuYan14}).

\subsection*{Structure of the paper} The rest of the paper is organized as follows. In Section \ref{S:equilibrium_simple}, we provide the model and establish existence of the equilibria under consideration. Section \ref{S:sigpx_compare_simple} develops quantitative analysis and qualitative discussion on information transmission and signal and price sensitivities. Section \ref{S:welfare} focuses on traders' utility in different equilibria, and Section \ref{S:demand_compare} examines equilibria structure regarding prices and risk allocation. Section \ref{S:conclude} concludes with a discussion on the model's predictions. An extension to a general multi-asset model is provided in Appendix \ref{AS:equilibrium_general_a} and all proofs are provided in Appendix \ref{AS:proofs}.

%---------------------------------------------%
\bigskip

\section{The Equilibrium}\label{S:equilibrium_simple}

We now present the model and construct the equilibrium. To both isolate the effects of price impact internalization and asymmetric information, and to keep the presentation/notation as simple as possible, we  consider a simplified model with only one risky tradeable asset whose payoff has zero mean and unit variance, and where the initial endowments are at Pareto optimality absent private information\footnote{Even though traders have CARA preferences, equilibrium quantities under price-impact depend on the initial endowments, as it is the insider's trade off her initial position which impacts equilibrium prices.   Assuming initial endowments are Pareto optimal absent private information turns off the channel where hedging demands based on the initial position affect equilibrium quantities, allowing us to isolate information and internalization effects. }. The single asset model, and all the main results, are generalized in Appendix \ref{AS:equilibrium_general_a} to multiple dimensions with general means and covariances.

\subsection*{Model and Uncertainty} The model has one period and all random quantities are defined on a probability space $(\Omega,\F,\prob)$, where the beliefs $\prob$ are common to all agents. The risky asset has terminal payoff $X\sim N(0,1)$ and positive supply $\Pi > 0$. The risk-less asset is in zero net supply, with price normalized to $1$.  Following the  literature\footnote{See \cite{subrahmanyam1991risk, spiegel1992informed, grossman1976efficiency, grossman1980impossibility} amongst many others.}, there is an insider $I$ who at time $0$ obtains a private signal $G$ which is a noisy version of $X$, taking the form
\begin{equation*}
    G = X + Z_I;\qquad Z_I  = \frac{1}{\sqrt{p_I}} \E_I,
\end{equation*}
where $\E_I \sim N(0,1)$ is independent of $X$ and $p_I > 0$ is the signal precision. There is also a mass of uninformed traders who do not receive a private signal, but who in the equilibrium established below, will receive a market signal through the time $0$ price. We assume the uninformed traders are price-takers, and following convention, we consider a representative agent $U$, hereafter called the uninformed trader. We assume both traders have exponential preferences with respective risk tolerances $\alpha_I$, $\alpha_U$. Lastly, there are liquidity providers (also called noise traders), denoted by $N$, with exogenous demand
\begin{equation*}
    Z_N = \frac{1}{\sqrt{p_N}}\E_N,
\end{equation*}
where $\E_N\sim N(0,1)$ is independent of both $X$ and $\E_I$.  $p_N$ measures the noise trader demand precision, and, as mentioned in \cite{KovViv14}, can be thought of as a measure of the volume of the price inelastic demand.  Traders $I$ and $U$ are endowed with (constant) share positions $\cbra{\pi_{i,0}}$ which are Pareto optimal absent private information\footnote{Since the uninformed agent, being a price taker, can be seen as a representative agent for a group of CARA traders, $\pi_{U,0}$ stands for their aggregate initial position and $\alpha_U$ denotes their aggregate risk tolerance. Also, one may allow the liquidity providers to have initial endowment $\pi_{N,0} \neq 0$, but this could just be absorbed into the supply $\Pi$. As such, we take $\pi_{N,0} = 0$.}
\begin{equation}\label{E:pareto_endow}
    \pi_{i,0} = \alpha_i \wh{\Pi} \quad i\in \cbra{I,U},\qquad \wh{\Pi} \dfn \frac{\Pi}{\alpha_I + \alpha_U}.
\end{equation}
By Pareto optimality of the initial endowments we mean there is a measure $\QQ_0$ equivalent to $\prob$ such that agents' marginal utility from endowed wealth is proportional to the density $d\QQ_0/d\prob$. Due to exponential preferences, the density takes the form
\begin{equation}\label{E:Q0_def}
\frac{d\QQ_0}{d\prob} = \frac{e^{-\wh{\Pi}X}}{\expvs{e^{-\wh{\Pi}X}}}.
\end{equation}

At time $0$ when the signal arrives, $I$ and $U$, using their respective information sets, choose positions $\pi_{I},\pi_{U}$ to take in the risky asset, financing this choice by trading in the riskless asset. As it is natural to present results for risk-aversion adjusted strategies (henceforth called a ``strategy''), we write
\begin{equation*}
\psi_i \dfn \frac{\pi_i}{\alpha_i};\qquad  i\in\cbra{I,U},
\end{equation*}
so the clearing condition is
\begin{equation}\label{E:cc_psi_simple}
\Pi = \alpha_I\wh{\psi}_I + \alpha_U \wh{\psi}_U + Z_N,
\end{equation}
and the risk-aversion adjusted terminal wealth is
\begin{equation}\label{E:tw_psi_simple}
\mathcal{W}^{\psi_i} \dfn \frac{1}{\alpha_i} \left(\pi_{i,0} p + \pi_{i}(X-p)\right)  =  \wh{\Pi}p + \psi_i(X-p);\qquad i\in\cbra{I,U}.
\end{equation}

%---------------------------------------------%

\subsection*{Price-impact equilibrium}\label{SS:eq_simple} Due to the exponential-Gaussian structure, we expect a linear-impact equilibrium. In other words, the insider perceives that if she changes her position from $\pi_{I,0} = \alpha_I \psi_{I,0}$ to  $\pi_I = \alpha_I \psi_I$, then the price will be an affine function of her trade combined with the noise trader demand,
\begin{equation}\label{E:impact_form_n}
p_{\iota}(\psi_I,Z_N) = V_{p,\iota} + M_{p,\iota} \left( \psi_I -\psi_{I,0} +  \frac{Z_N}{\alpha_I}\right),
\end{equation}
for constants $V_{p,\iota}, M_{p,\iota}$ that are determined in equilibrium,  and where throughout we use the subscript ``$\iota$'' to stand for ``impact''.  Following \cite{rochet_vila}, we assume the insider can see both her private signal and, for a given trade $\psi$, the price $p_{\iota}$. This implies the noise trader demand $Z_N$ is revealed to the insider as well.\footnote{This is in contrast to  \cite{spiegel1992informed} and \cite{subrahmanyam1991risk} and leads to a different equilibrium, where the insider must submit her demand order prior to seeing the price.}  As such, her acceptable policies $\A_I$ are functions $\psi = \psi(G,Z_N)$ of the private signal and noise,\footnote{$\A_I$ is formally defined in equation \eqref{eq:A_I_a} below, as there is a technical restriction on $\A_I$ which is needed a-priori, but which always holds in equilibrium.} and for fixed $V_{p,\iota},M_{p,\iota}$, the insider's optimal demand solves
\begin{equation*}
\begin{split}
&\inf_{\psi \in \A_I}\condexpvs{e^{-\psi_{I,0} p_{\iota}\left(\psi, Z_N\right) -\psi\left(X-p_{\iota}\left(\psi, Z_N\right)\right)}}{\sigma(G,Z_N)}.
\end{split}
\end{equation*}
Due to linear impact and the exponential-Gaussian structure, the insider's optimal policy $\wh{\psi}_{I,\iota}$ is affine in the private signal $G$ and noise trader demand $Z_N$. Therefore, with an eye towards \eqref{E:impact_form_n}, we may write
\begin{equation*}
    \wh{\psi}_{I,\iota}(g,z) - \psi_{I,0} + \frac{z}{\alpha_I} = V_{I,\iota} + M_{I,\iota}\left(g + \Lambda_{\iota} z\right),
\end{equation*}
for certain constants $V_{I,\iota},M_{I,\iota}$ and most importantly $\Lambda_{\iota}$ (explicitly given in \eqref{E:pi_signal_simple_n} below), each of which depends upon the pricing coefficients $V_{p,\iota},M_{p,\iota}$.  As the quantity on the left side above is publicly observable, this implies the market receives the signal
\begin{equation*}
    H_{\iota} = G + \Lambda_{\iota} Z_N,
\end{equation*}
and from \eqref{E:impact_form_n} we may view the price as $p_{\iota}(H_{\iota})$. Therefore, the uninformed agent's acceptable policies $\A_U$ are functions $\psi = \psi(H_{\iota})$, and as the uninformed is a price taker (and also because the risk aversion adjusted initial wealth $\psi_{U,0} p_{\iota}(H_{\iota})$ factors out), his optimization problem is
\begin{equation*}
\begin{split}
&\inf_{\psi \in \A_U}\condexpvs{e^{- \psi\left(X-p_{\iota}\left(H_{\iota}\right)\right)}}{\sigma(H_{\iota})}.
\end{split}
\end{equation*}
Similarly to the insider, the uninformed agent's optimal policy function is affine in the public signal $H_{\iota}$, with coefficients that also depend on $V_{p,\iota},M_{p,\iota}$.  Therefore, if we enforce the market clearing condition \eqref{E:cc_psi_simple}, we end up with an equation of the form (which must hold with probability one)
\begin{equation*}
    0 = V_{H,\iota}(V_{p,\iota}, M_{p,\iota}) + M_{H,\iota}(V_{p,\iota}, M_{p,\iota})\times  H_{\iota},
\end{equation*}
where we have made explicit the dependence of the coefficients on $V_{p,\iota},M_{p,\iota}$.  Thus, provided one can find $V_{p,\iota},M_{p,\iota}$ such that $0 = V_{H,\iota} = M_{H,\iota}$ there exists a linear price-impact equilibrium.

Similarly to \cite{subrahmanyam1991risk} where equilibrium quantities were governed by solutions to a fifth order equation, in our model equilibrium quantities are governed by positive solutions to the cubic equation
\begin{equation}\label{E:cubic_alt_n}
0 = (1+y)^2\left(1-\frac{1-\consb}{\consb(1+p_I)}y\right) + \frac{\consa p_I}{\consb}\left((1-\consb)y + 1\right),
\end{equation}
where
\begin{equation}\label{E:lambda_alpha_R_def_n}
\consa \dfn \alpha_I^2 p_N;\qquad \consb \dfn \frac{\alpha_I}{\alpha_I+\alpha_U},
\end{equation}
are the precision of $Z_N/\alpha_I$ and the insider's proportion of the total risk tolerance respectively.  As a first step we have the following.
\begin{proposition}\label{prop:one_d_n}
There exists a unique strictly positive solution $\wh{\yvals}$ to \eqref{E:cubic_alt_n}.
\end{proposition}
Given this proposition, our main result establishes equilibrium. To state it, recall the measure $\QQ_0$  from \eqref{E:Q0_def} and define
\begin{equation}\label{E:new_p_nsn_simple_n}
\pnsn \dfn \expv{\QQ_0}{}{X} = -\wh{\Pi},
\end{equation}
as the equilibrium price absent private information. By expressing quantities in terms of $\pnsn$ we provide an intuitive way to gauge how private information and price impact alter prices. 

\begin{theorem}\label{thm:pi_simple_n}
$\wh{\yvals}$ from Proposition \ref{prop:one_d_n} induces a price-impact equilibrium. The market signal is
\begin{equation}\label{E:pi_signal_simple_n}
H_{\iota} = G + \Lambda_{\iota} Z_N,\qquad \Lambda_{\iota} = \frac{1+\wh{y}}{\alpha_I p_I}.
\end{equation}
$H_{\iota}$ is of the same form as the insider signal $G$, except with lower precision
\begin{equation}\label{E:PU_pi_simple_n}
p_{U,\iota} =  \left(\recip{p_I} + \frac{1}{p_N\Lambda^2_{\iota}}\right)^{-1} =  p_I \times \frac{\consa p_I }{(1+\wh{\yvals})^2 + \consa p_I}.
\end{equation}
The equilibrium price is $p_{\iota}(H_{\iota})$ for the price function
\begin{equation}\label{E:pi_price_new_simple_n}
    \begin{split}
    p_{\iota}(h_{\iota}) &=\pnsn + \frac{p_I \wh{\yvals}}{(1+p_I)(1+2\wh{\yvals})}\left(h_{\iota} - \pnsn\right).
    \end{split}
\end{equation}
Writing $H_{\iota}  = h_{\iota}(G,Z_N)$, the insider has optimal policy function
\begin{equation}\label{E:pi_psiI_simple_n}
\begin{split}
    \wh{\psi}_{I,\iota}(g,z) &= \frac{p_I}{1+\wh{\yvals}}g - \frac{1+p_I}{1+\wh{\yvals}}p_{\iota}(h_{\iota}(g,z)) - \frac{\wh{\yvals}}{1+\wh{\yvals}}\pnsn.
\end{split}
\end{equation}
The uninformed agent has optimal policy function
\begin{equation}\label{E:pi_psiU_simple_n}
    \wh{\psi}_{U,\iota}(h_{\iota}) = p_{U,\iota} h_{\iota} - (1+p_{U,\iota})p_{\iota}(h_{\iota}).
\end{equation}
Lastly, the coefficients in \eqref{E:impact_form_n} are $M_{p,\iota} = \wh{\yvals}/(1+p_I)$ and $V_{p,\iota} = -\wh{\Pi} = \pnsn$. 
\end{theorem}

\subsection*{Price-taking equilibrium} For comparison purposes,  herein we consider the case where all agents are price takers. As this result is well known (see \cite{grossman1980impossibility}), we summarize the equilibrium structure in the following proposition.  To state it, assume there is a market signal $H$ revealed through the time $0$ price $p=p(H)$, and both traders take $p(H)$ as given. The insider has time 0 information $\sigma(H,G)$ while the uninformed trader uses $\sigma(H)$.  Using \eqref{E:tw_psi_simple}, the insider and uninformed trader's optimal investment problems are respectively
\begin{equation}\label{E:pt_vf_I_n}
    \inf_{\psi \in \sigma(G,H)} \condexpvs{e^{-\mathcal{W}^{\psi}}}{\sigma(G,H)};\qquad \inf_{\psi \in \sigma(H)} \condexpvs{e^{-\mathcal{W}^{\psi}}}{\sigma(H)}.
\end{equation}
We say $(H,p(H))$ is a price-taking equilibrium if the clearing condition \eqref{E:cc_psi_simple} holds for the optimal policies.  In the proposition below, note the similarities to Theorem \ref{thm:pi_simple_n}. This will form the basis of our comparison results in the next section.

\begin{proposition}\label{P:pt_equilibrium_simple_n}
There is a price-taking equilibrium. The market signal is
\begin{equation}\label{E:pt_signal_simple_n}
H \dfn G + \Lambda Z_N,\qquad \Lambda = \frac{1}{\alpha_I p_I}.
\end{equation}
$H$ is of the same form as $G$, but  with lower precision (see \eqref{E:lambda_alpha_R_def_n})
\begin{equation}\label{E:pt_PU_simple_n}
p_U = p_I \times \frac{\consa p_I}{1+\consa p_I}.
\end{equation}
The equilibrium price is $p=p(H)$ for the price function (recall \eqref{E:new_p_nsn_simple_n})
\begin{equation}\label{E:pt_price_simple_n}
\begin{split}
p(h) &\dfn \pnsn +  \frac{\alpha_I p_I +\alpha_U p_U}{\alpha_I(1+p_I) + \alpha_U(1+p_U)}(h-\pnsn).
\end{split}
\end{equation}
Writing $H = h(G,Z_N)$, the insider has optimal policy function
\begin{equation}\label{E:pt_psiI_simple_n}
    \begin{split}
        \wh{\psi}_I(g,z) &= p_I g - (1+p_I)p(h(g,z)).
    \end{split}
\end{equation}
The uninformed agent has optimal policy function
\begin{equation}\label{E:pt_psiU_simple_n}
    \wh{\psi}_U(h) = p_U h - (1+p_U)p(h).
\end{equation}
\end{proposition}

\subsection*{No private signal equilibria}  Lastly, we provide results absent private information, where the only uncertainty at time $0$ is the noise trader demand $Z_N$. We turn off the asymmetric information channel to analyze effects due solely to internalization of price impact.\footnote{Qualitatively, the no signal limit corresponds to when there is a market maker who is capable of moving prices, but who is not privately informed about the asset's terminal payoff.  She aims to move prices against a mass of (small) uninformed traders in a way to maximize her utility.}  Here, it turns out that the no-signal equilibria (in both the price-impact and price-taking cases) coincide with the previously established equilibria in the limit $p_I \to 0$. The resultant prices and optimal positions are summarized in the following proposition, the proof of which is given in Appendix \ref{AS:equilibrium_a}.

\begin{proposition}\label{P:no_signal_simple_n}
No-signal equilibria correspond to $p_I = 0$. In the price-taking case, the equilibrium price and optimal positions are $p_{ns}(Z_N)$ and $\wh{\psi}_{ns,I}(Z_N) = \wh{\psi}_{ns,U}(Z_N) = \wh{\psi}_{ns}(Z_N)$ where (recall \eqref{E:pareto_endow}, \eqref{E:lambda_alpha_R_def_n} and \eqref{E:new_p_nsn_simple_n})
\begin{equation}\label{E:no_signal_competi_simple_n}
p_{ns}(z)=\pnsn + \consb\frac{z}{\alpha_I};\qquad \wh{\psi}_{ns}(z) = \wh{\Pi} - \consb\frac{z}{\alpha_I}.
\end{equation}
In the price-impact case, the equilibrium price is $p_{ns,\iota}(Z_N)$ and the optimal policies are $\wh{\psi}_{ns,I,\iota}(Z_N)$  and $\wh{\psi}_{ns,U,\iota}(Z_N)$, where
\begin{equation*}
    \begin{split}
    p_{ns,\iota}(z) &= \pnsn + \frac{\consb}{1-\consb^2}\frac{z}{\alpha_I};\quad \wh{\psi}_{ns,I,\iota}(z) = \wh{\Pi} - \frac{\consb}{1+\consb}\frac{z}{\alpha_I};\quad \wh{\psi}_{ns,U,\iota}(z) = \wh{\Pi} - \frac{\consb}{1-\consb^2}\frac{z}{\alpha_I}.
    \end{split}
\end{equation*}

\end{proposition}

\subsection*{A comment on the linearity of price impact}

Our assumption of linear price impact is motivated by the price-taking case. Indeed, using \eqref{E:pt_price_simple_n}, \eqref{E:pt_psiI_simple_n} one can show
\begin{equation}\label{E:strat_map_simple}
\begin{split}
p(h(g,z))&=  \pnsn + \frac{\alpha_Ip_I + \alpha_U p_U}{\alpha_U(p_I - p_U)}\left(\wh{\psi}_I(g,z) - \psi_{I,0} +\frac{z}{\alpha_I}\right).
\end{split}
\end{equation}
This is the reverse combined demand function at equilibrium, and indicates linear price impact. Indeed, even though the insider does not internalize impact in the price-taking case, in equilibrium the price is linearly impacted by her trade, combined with the noise trader's demand.  The price takes the form \eqref{E:impact_form_n}, where
\begin{equation*}
    V_{p} = \pnsn = -\wh{\Pi}, \qquad M_{p} = \frac{\alpha_Ip_I + \alpha_U p_U}{\alpha_U(p_I - p_U)}.
\end{equation*}
As discussed in the introduction, as long as the insider internalizes her price impact and the uninformed trader takes this into account, the price-taking and price-impact equilibria cannot coincide. In fact, even if the insider submits the demand which is optimal in the price-taking equilibrium,  if she internalizes price impact, the market will not equilibrate to the price-taking equilibrium price. This would be the case if the uninformed trader did not perceive the change in market signal precision due to the insider's demand (i.e., if he ignored the insider's internalization of the price impact and assumed $p_{U,\iota}(\yvals)=p_U$ for all $\yvals$.)

On the other hand, there is a  $\yvals^* > 0$ such that  $M_{p,\iota}(\yvals^*) = \yvals^*/(1+p_I)$ coincides with $M_p$ from above. For this $M_{p,\iota}(\yvals^*)$, if the insider used the price-taking optimal demand $\wh{\psi}_I$ from \eqref{E:pt_psiI_simple_n}, it would reveal the same signal to the uninformed trader as in the price-taking equilibrium. This would lead to the same uninformed trader's demand and hence the same clearing price as in price-taking equilibrium. However, when $M_{p,\iota} = M(\yvals^*)$ and $V_{p,\iota}=-\wh{\Pi}$,  $\wh{\psi}_I$ from \eqref{E:pt_psiI_simple_n} is not optimal for the insider in the price-impact model, hence the price-taking and price-impact equilibria cannot be the same.

%---------------------%

%-----------------------------------------%

\section{Comparison Analysis: Signals and Price Sensitivity}\label{S:sigpx_compare_simple}

In this section, we compare the public signals and price sensitivity with respect to signals of the two equilibria. As in the introduction, we label the price-impact equilibria as ``PI'' and the price-taking equilibria as ``PT''.  We show the PI public signal is of a worse quality, and prices are less responsive to not only the market and insider signals, but also to the publicly observable (risk-tolerance weighted) insider's and noise trader residual demand $\wh{\psi}_I - \psi_{I,0} + Z_N/\alpha_I$. Thus, the main message of this section is
\begin{quote}
    \textit{By assuming the insider is a price taker, one overestimates the quality of the public signal and the reactivity of equilibrium prices.}
\end{quote}
Throughout, we include the subscript $\iota$ when describing any quantity obtained internalizing price impact. Proofs are in Appendix \ref{AS:sigpx_compare}, and we collect $p_{U,\iota}$ from \eqref{E:PU_pi_simple_n} and $p_U$ from \eqref{E:pt_PU_simple_n}
\begin{equation}\label{E:constants_nice_simple}
\begin{split}
p_{U,\iota} = \frac{\kappa p_I^2}{(1+\wh{y})^2 + \kappa p_I};\qquad p_U = \frac{\kappa p_I^2}{1+\kappa p_I}.
\end{split}
\end{equation}

\subsection*{Signal quality}
As we have seen, in both the PI and PT equilibria a signal of the form ``$X + \textrm{Noise}$'' is communicated to market. It is natural to ask which signal is of a higher quality, or even more pointedly, is the public signal less informative in the presence of price impact? To address these questions we write the market signals as functions of the insider signal $G$ and noisy demand $Z_N$, and using \eqref{E:pi_signal_simple_n}, \eqref{E:pt_signal_simple_n}  we obtain
\begin{equation}\label{E:signals_together}
\begin{split}
h_{\iota}(g,z) &= g + \frac{1+\wh{y}}{p_I}\frac{z}{\alpha_I};\qquad h(g,z) = g + \frac{1}{p_I} \frac{z}{\alpha_I}.
\end{split}
\end{equation}
From Theorem \ref{thm:pi_simple_n} we know $\wh{y} > 0$, which from \eqref{E:constants_nice_simple} implies
\begin{proposition}\label{prop:impact_signal_noise}
The market signal is noisier in the PI equilibria: $p_U>p_{U,\iota}$.
\end{proposition}
That the public signal is less informative under price impact is associated with the way the uninformed trader determines his demand. Indeed, because the uninformed trader accounts for the insider's internalization of price impact, in his optimization problem he considers $p_{U,\iota}$ instead of $p_U$ (compare the demand functions \eqref{E:pi_psiU_simple_n} and \eqref{E:pt_psiU_simple_n}). In other words, he recognizes the insider reveals a wangled signal, and responds with a less elastic demand function.

\subsection*{Price reactivity}  In both the PI and PT equilibria, prices are affine functions of the respective public signals. However, the coefficients in the functions differ. This leads one to ask whether price impact increases or decreases the sensitivity of prices with respect to public signaling. To  answer this, we recall the PI price $p_{\iota}$ from \eqref{E:pi_price_new_simple_n} and re-express the PT price \eqref{E:pt_price_simple_n} using the notation of \eqref{E:lambda_alpha_R_def_n}.

\begin{proposition}\label{prop:prices_simple}
The pricing functions take the form
\begin{equation}\label{E:prices_simple}
\begin{split}
p_{\iota}(h_{\iota}) &= \pnsn + \frac{p_I\wh{y}}{(1+p_I)(1+2\wh{y})}\left(h_{\iota}-\pnsn\right);\\
p(h) &= \pnsn + \frac{p_I(\consa p_I +\consb)}{1-\consb  + (1 + p_I)(\consa p_I + \consb)}\left(h-\pnsn\right).
\end{split}
\end{equation}
\end{proposition}

Given this, and in view of \eqref{E:pi_signal_simple_n}, \eqref{E:pt_signal_simple_n}, define the slopes
\begin{equation}\label{E:mg_alt_compare}
\begin{split}
m_{g,\iota} = \frac{p_I\wh{y}}{(1+p_I)(1+2\wh{y})},\qquad m_g = \frac{p_I(\consa p_I +\consb)}{1-\consb  + (1 + p_I)(\consa p_I + \consb)}.
\end{split}
\end{equation}
The following proposition shows prices are always more reactive to the insider, and hence to the market signal, when the insider does not internalize price impact. This is directly linked with the lower elasticity of the uninformed trader's demand function due to price impact, and is an endogenously derived outcome. The insider is motivated to make the public signal noisier, which makes uninformed trader less elastic and yields prices which are less sensitive to the public signal.

\begin{proposition}\label{prop:price_impact_signal} The equilibrium price is less sensitive to the market signal in the PI equilibria: $m_{g,\iota}<m_{g}$.
\end{proposition}

We conclude this discussion with the price reactivity with respect to the publicly observable (weighted risk-tolerance adjusted) combined demand
\begin{equation*}
\wh{\chi}_{\iota} \dfn \wh{\psi}_{I,\iota}(G,Z_N) -\psi_{I,0} + \frac{1}{\alpha_I}Z_N,\qquad \wh{\chi} \dfn \wh{\psi}_{I}(G,Z_N) -\psi_{I,0} + \frac{1}{\alpha_I}Z_N.
\end{equation*}
Using Theorem \ref{thm:pi_simple_n}, Proposition \ref{P:pt_equilibrium_simple_n} and \eqref{E:lambda_alpha_R_def_n} one can show prices are affine in the combined demand with respective slopes
\begin{equation*}
m_{\wh{\chi},\iota} = \frac{\wh{y}}{1+p_I},\qquad m_{\wh{\chi}} = \frac{\consb+\consa p_I}{1-\consb}.
\end{equation*}
As expected from the preceding analysis, when the insider internalizes her impact, prices are less sensitive to the publicly observable combined demand (similarly to the public signal). Again, we stress this is an endogenous outcome, arising from the insider's strategy when she internalizes her price impact. The next proposition formally states this result.

\begin{proposition}\label{prop:price_impact_dd} The equilibrium prices are less sensitive to the publicly observable combined demand in the PI equilibria: $m_{\wh{\chi},\iota}< m_{\wh{\chi}}$.
\end{proposition}
%-----------------------------------------%

\section{Indirect Utility}\label{S:welfare} This section is dedicated to analyzing the traders' indirect utility (called ``utility'' or ``value''). We address how the insider's signal quality relates to her utility, and how price-impact internalization effects all agents' utility. We again label the price impact equilibrium PI and the price taking equilibrium PT.

We define utility at both the ex-ante level (i.e.~at time $0-$, prior to signal revelation) and interim level (at time $0$, after the signal revelation) in terms of certainty equivalents.  Utility will always be computed using the overall wealth in \eqref{E:tw_psi_simple}. Given this, we denote by $\wh{\mathcal{W}}_{I,\iota}, \wh{\mathcal{W}}_{U,\iota}$ the optimal terminal wealths in the PI equilibrium, and $\wh{\mathcal{W}}_I, \wh{\mathcal{W}}_U$ those in the PT case. Then, for $k\in\cbra{\ , \iota}$ the corresponding interim certainty equivalents are
\begin{equation*}
    \begin{split}
\textrm{CE}^{I}_{0,k} &= -\alpha_I \log\left(\condexpvs{e^{-\frac{1}{\alpha_I}\wh{\mathcal{W}}_{I,k}}}{\sigma(G,H_{k})}\right);\quad \textrm{CE}^{U}_{0,k} = -\alpha_U \log\left(\condexpvs{e^{-\frac{1}{\alpha_U} \wh{\mathcal{W}}_{U,k}}}{\sigma(H_{k})}\right),
    \end{split}
\end{equation*}
while the ex-ante certainty equivalents are
\begin{equation*}
\textrm{CE}^{j}_{0-,k} = -\alpha_j\log\left(\expvs{e^{-\frac{1}{\alpha_j} \wh{\mathcal{W}}_{j,k}}}\right);\qquad j\in\cbra{I,U}.
\end{equation*}

We find that (i) in the PI equilibrium, insider ex-ante utility is always strictly increasing in the precision $p_I$; (ii) insider ex-ante utility is not always higher in the PI equilibrium; and (iii) absent private information, insider utility is always higher in the PI equilibrium, and remarkably this holds at the interim level. In fact, when the insider and uninformed traders have the same risk tolerance we can order interim utility 
\begin{equation*}
    \textrm{U(PI)} > \textrm{I(PI)} > \textrm{U(PT)} = \textrm{I(PT)},
\end{equation*}
so the uninformed trader's utility exceeds the insider's utility in the PI case.  We use the notation in \eqref{E:lambda_alpha_R_def_n} and, Proposition \ref{prop: exante_I_pipt} aside, all proofs of this section are in Appendix \ref{AS:welfare_2}.

\subsection*{Certainty Equivalents} We start calculating the certainty equivalents. Propositions \ref{P:PI_CEs_a} and \ref{P:PT_CEs_a} (see also Remark \ref{R:nice_way}) below compute ex-ante utilities in the PI and PT equilibria under the general model of Appendix \ref{AS:equilibrium_general_a}.  To state the results in the model of Section \ref{S:equilibrium_simple} define
\begin{equation}\label{E:new_ce_nsn}
 \textrm{CE}_{nsn}^{i} \dfn - \frac{\alpha_i}{2}\wh{\Pi}^2,\qquad i\in\cbra{I,U},
\end{equation}
as the certainty equivalents absent private information for the allocations in \eqref{E:pareto_endow}.

\begin{proposition}\label{prop: welfare_nice_endow}
In the PI equilibrium, with $\wh{y}$ from Proposition \ref{prop:one_d_n}
\begin{equation*}
\begin{split}
\cexaio{I} &=  \textrm{CE}_{nsn}^{I}  + \frac{\alpha_I}{2}\log\left(1 + \frac{\consa p_I(1+p_I) + \wh{y}^2}{\consa (1+p_I)(1+2\wh{y})}\right),\\
\cexaio{U} &=  \textrm{CE}_{nsn}^{U}  + \frac{\alpha_U}{2} \log\left(1 + \frac{\consb^2(\consa p_I+ (1+\wh{y})^2)}{(1-\consb)^2\consa (1+2\wh{y})^2} \right).
\end{split}
\end{equation*}
In the PT equilibrium
\begin{equation*}
\begin{split}
\cexao{I} &= \textrm{CE}_{nsn}^{I}  + \frac{\alpha_I}{2}\log\left(1 + \frac{(1-\consb)^2\consa p + (1+p_I)(\consb+\consa p_I)^2}{\consa(1 + \consb p_I + \consa p_I (1+p_I))^2}\right),\\
\cexao{U} &= \textrm{CE}_{nsn}^{U}  + \frac{\alpha_U}{2}\log\left(1 +\frac{\consb^2(1+\consa p_I)}{\consa(1 + \consb p_I + \consa p_I(1+p_I))^2}\right).
\end{split}
\end{equation*}
\end{proposition}

\subsection*{Utility and the insider's signal precision} Though not modeled, it presumably costs effort, time and/or money for the insider to obtain and refine her signal. In fact, one of the central questions in the literature (see among others \cite{KacPag19} and \cite{KacNosSun23}) is whether traders who have the ability/resources to obtain a private signal should actually pay the cost and obtain it. To answer this question, one should examine whether the benefits of the  private signal (as measured by utility) are increasing with respect to the quality of a signal (as measured by signal precision). This is of course connected with the cost, in that a better signal is normally linked to a higher cost.

In the PI equilibrium, using Proposition \ref{prop: welfare_nice_endow} for fixed $\kappa,\lambda$, it suffices to study the map
\begin{equation}\label{E:fpimon}
    p_I \to \phi_{\iota}(p_I) \dfn \frac{\consa p_I(1+p_I) + \wh{y}(p_I)^2}{\consa (1+p_I)(1+2\wh{y}(p_I))},
\end{equation}
where $\wh{y} = \wh{y}(p_I)$ is the unique positive solution of \eqref{E:cubic_alt_n}. Numerically, this is seen to be increasing in $p_I$ by randomly sampling $\consa > 0$, $\consb \in (0,1)$, solving for $\wh{y}(p_I)$, and then plotting $p_I \to \phi_{\iota}(p_I)$.  However, we offer an analytic proof in the following theorem.
\begin{theorem}\label{T:pimono}
    For fixed $\consa > 0$ and $\consb \in (0,1)$ the map $\phi_{\iota}$ defined in \eqref{E:fpimon} is strictly increasing in $p_I$.  Therefore, $\cexaio{I}$ is strictly increasing in the precision $p_I$.
\end{theorem}

Alternatively, in the PT equilibrium, one must study the map
\begin{equation*}
    p_I \to \phi(p_I) \dfn \frac{(1-\consb)^2\consa p_I + (1+p_I)(\consb+\consa p_I)^2}{\consa(1 + \consb p_I + \consa p_I (1+p_I))^2}.
\end{equation*}
As $\phi(0) = \consb^2/\consa$ and $\phi(\infty) = 0$, this map is clearly not increasing. In fact, it is not monotonic because
\begin{equation*}
    \phi'(0) = 1-\consb^2 + \frac{\consb^2}{\consa}(1-2\consb).
\end{equation*}
When $\consb \leq 1/2$ (equivalently $\alpha_I \leq \alpha_U$ so the insider is less risk tolerant than the uninformed traders), $\phi$ is increasing at $0$. However, when the insider is relatively more risk tolerant ($\consb > 1/2$), $\phi$ will be decreasing at $0$ for $\consa=\alpha_I^2 p_N$ small enough (which can happen if the noise trader variance/volume is very large).

Interestingly, it is possible for the certainty equivalents with and without price impact to \textit{have the opposite monotonicity with respect to the signal precision}. This is pictured in Figure \ref{F:CE_vs_I}. We conclude that if the insider does not internalize her price impact, her utility may not increase in the signal's quality. But  when she internalizes her price impact, it is always beneficial for her to improve the quality of her private signal, as long as the cost of such improvement does not outweigh the corresponding utility increase.

\begin{figure}
\begin{center}
\includegraphics[height=4.50cm,width=10cm]{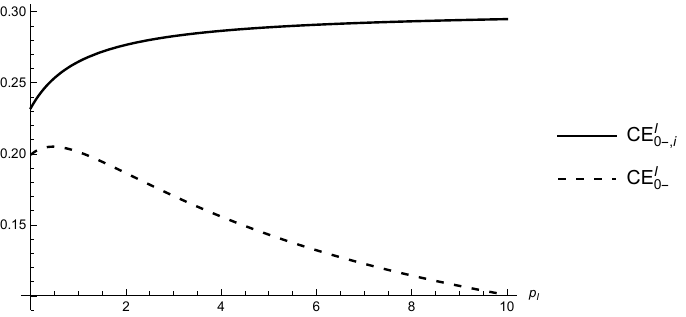}
\end{center}
\caption{Plot comparing $\cexaio{I}$ and $\cexao{I}$ as a function of $p_I$. Parameters are $\alpha_I = \alpha_U
 = .3, \mu_X = .5, P_X = 1, p_N = 1, \Pi = 0$.}
\label{F:CE_vs_I}
\end{figure}

\subsection*{PI and PT utility comparison}

Next, we focus on how price impact internalization affects both informed and uninformed trader's utility. In particular, we examine whether the internalization of price impact implies higher utility for the insider and the uninformed trader, when compared to the price taking case. Using Proposition \ref{prop: welfare_nice_endow} we readily get the relations
\begin{equation}\label{E:PI_better_IU}
\begin{split}
    \cexaio{I} \geq \cexao{I} &\quad \Longleftrightarrow \quad \frac{\consa p_I(1+p_I) + \wh{y}^2}{(1+p_I)(1+2\wh{y})} \geq \frac{(1-\consb)^2\consa p_I + (1+p_I)(\consb+\consa p_I)^2}{(1 + \consb p_I + \consa p_I (1+p_I))^2};\\
    \cexaio{U} \geq \cexao{U} &\quad \Longleftrightarrow \quad \frac{\consa p_I+ (1+\wh{y})^2}{(1-\consb)^2 (1+2\wh{y})^2} \geq \frac{1+\consa p_I}{(1 + \consb p_I + \consa p_I(1+p_I))^2}.
\end{split}
\end{equation}
Our first result shows that insider utility need not increase when she internalizes her price impact.

\begin{proposition}\label{prop: exante_I_pipt}
Both $\cexaio{I} > \cexao{I}$ and $\cexaio{I} < \cexao{I}$ are possible.
\end{proposition}

\begin{proof}[Proof of Proposition \ref{prop: exante_I_pipt}]
Numerically, this is demonstrated in Figure \ref{F:I_welfare_comparision}. Indeed, in the joint combination of high $\alpha_U$ (uninformed close to risk neutrality) and low-to-moderate $p_I \in (0,2)$ (modest signal quality) utility may decrease.  Analytically, this will be shown in Proposition \ref{P:gamma_I_asympt_n}.
\end{proof}

\begin{figure}[ht!]
\begin{center}
\includegraphics[height=4.0cm,width=6cm]{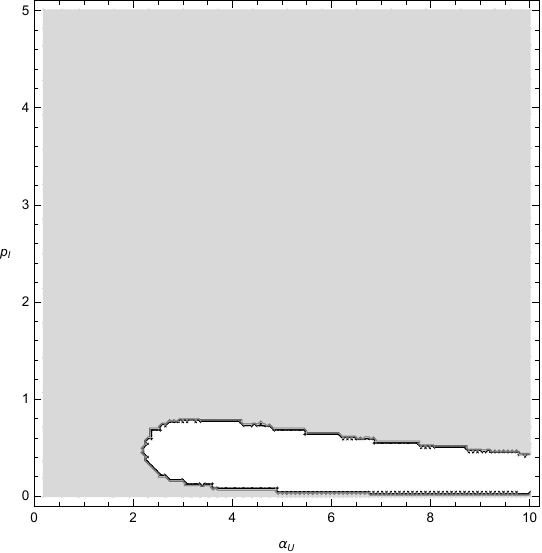}\ \includegraphics[height=4.0cm,width=6cm]{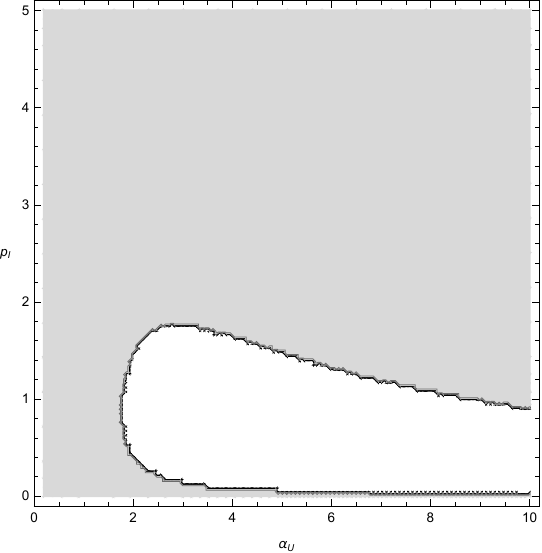}\\ \includegraphics[height=4.0cm,width=6cm]{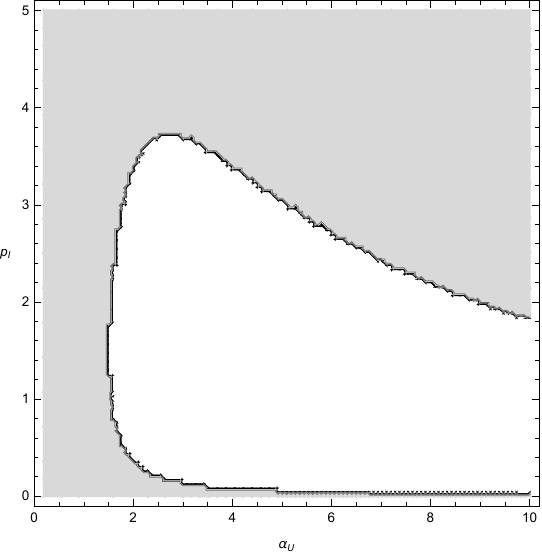}
\end{center}
\caption{Plot comparing $\cexaio{I}$ and $\cexao{I}$ as a function of $\alpha_U$ (x-axis) and $p_I$ (y-axis) for $p_N = 1$ and $\alpha_I = .2$ (upper left), $\alpha_I = .1$ (upper right), $\alpha_I = .05$ (lower). The shaded region is where $\cexaio{I} > \cexao{I}$. The white region is where $\cexaio{I} < \cexao{I}$}
\label{F:I_welfare_comparision}
\end{figure}

While no uniform statement can be made for the insider, our next result shows  for the uninformed trader, at the ex-ante level the insider's internalization of price impact is always beneficial.

\begin{proposition}\label{prop: exante_U_pipt}
$\cexaio{U} \geq \cexao{U}$.
\end{proposition}

Economic intuition for these results is given in Section \ref{S:demand_compare}.  Below, we provide explicit comparisons when investor risk aversions become very large or small.

\subsection*{Risk tolerance asymptotics}

It is rather complicated to precisely describe the parameter set that characterizes the order of insider's certainty equivalents $\cbra{(\consa, p_I, \consb) \such \cexaio{I} \geq \cexao{I}}$. However,  Figure \ref{F:I_welfare_comparision} suggests the situation may clarify if we consider asymptotics with respect to traders' risk tolerances, especially as $\alpha_I \to 0$. In the following proposition, we first consider when $\alpha_I = \alpha_U$ and then take the limit as the insider's risk tolerance goes to $0$.

\begin{proposition}\label{P:gamma_I_asympt_n}
Fix $p_I,p_N$. If $\alpha_I=\alpha_U$, then $\cexaio{I} \geq \cexao{I}$. However, for $\alpha_U$ fixed we have 
\begin{equation*}
    \lim_{\alpha_I\to 0} \frac{\cexaio{I}-\cexao{I}}{\alpha_I^3}  = \frac{d}{2}\frac{(1+p_I)^2(1+p_I-\alpha_U^2 p_I p_N)}{\alpha_U^2(\alpha_U^2p_I p_N + 1+p_I)}.
\end{equation*}
Thus, for small $\alpha_I$,  $\cexaio{I}\geq \cexao{I}$ if and only if $1+p_I - \alpha_U^2 p_I p_N > 0$.
\end{proposition}

Proposition \ref{P:gamma_I_asympt_n} implies the risk aversion is a crucial parameter for the certainty equivalents' comparisons. Indeed, when traders have the same risk aversion internalizing price impact always increases traders' utility. On the other hand, fixing $\alpha_U,p_N$, insider utility may be lower than in the price-taking case when she is very risk averse and the following structure condition holds
\begin{equation}\label{E:less_CE_in_PI}
\alpha_U^2p_N > \frac{1}{p_I}+1.
\end{equation}
For example, this condition holds if the uninformed trader is sufficiently risk tolerant, or/and the noise traders' demand (approximated by its variance) is sufficiently low.

\subsection*{Utility absent private information} We conclude our utility analysis proving that if one turns off the asymmetric information channel, and focuses solely on effects due to internalizing of price impact, then insider utility increases at the interim level. As previously mentioned, this setting corresponds to when a (large) risk averse agent acts strategically, even when she does not have access to a private signal, and when the uninformed traders represent the mass of all other rational agents who do not possess  market power. In sum, we are measuring the effects heterogeneity in the internalization of price impact.

Recall from Proposition \ref{P:no_signal_simple_n} that the no-signal equilibrium corresponds to taking $p_I \to 0$.  Using Proposition \ref{prop: nosigint} in Appendix \ref{AS:welfare_2}, we obtain the following proposition.

\begin{proposition}\label{prop: nosigint_simple}
As $p_I \to 0$ we obtain the almost sure limits for the insider
\begin{equation*}
\begin{split}
\lim_{p_I\to 0} \ceinto{I}(G,Z_N) &=  \textrm{CE}_{nsn}^{I}  + \frac{\consb^2}{2\alpha_I}Z_N^2,\quad \lim_{p_I\to 0} \ceintio{I}(G,Z_N) = \textrm{CE}_{nsn}^{I}  + \frac{\consb^2}{2(1-\consb^2)\alpha_I}Z_N^2.
\end{split}
\end{equation*}
Therefore, insider interim utility always increases when internalizing price impact. For the uninformed trader we obtain almost surely
\begin{equation*}
\begin{split}
\lim_{p_I\to 0} \ceinto{U}(H) &= \textrm{CE}_{nsn}^{I}  + \frac{\alpha_U\consb^2}{2\alpha_I^2}Z_N^2,\quad \lim_{p_I\to 0} \ceintio{U}(H_{\iota}) =  \textrm{CE}_{nsn}^{I}  + \frac{\alpha_U\consb^2}{2\alpha_I^2(1-\consb^2)^2}Z_N^2,
\end{split}
\end{equation*}
and hence interim utility also always increases when the insider internalizes price impact.
\end{proposition}

\begin{remark}\label{rem: nosigint}
Two quite surprising consequences stem from the above proposition. First, that any utility gain the insider obtains in PT equilibrium over the PI equilibrium is attributable solely to the presence of the private signal, as when there is no private signal, internalization of price impact is always beneficial. Second, assuming no private signal and that $I,U$ have the same risk tolerance, (which implies $\consb = 1/2$) we have the almost sure order of interim certainty equivalents
\begin{equation*}
    \ceintio{U} > \ceintio{I} > \ceinto{U} =\ceinto{I}.
\end{equation*}
Amazingly, not only it is better for the uninformed trader when the insider internalizes her price impact, the uninformed trader's utility actually exceeds that of the insider. We explain the mechanism behind these predictions in the next section.
\end{remark}

\section{Equilibrium Structure}\label{S:demand_compare}

In this section, we analyze and compare the equilibrium quantities (allocation and prices) in order to infer the economic intuition behind the model's predictions on the effects of price-impact internalization (``internalization'') and insider's private signal (``private information''). As will be shown, internalization and private information may have competing affects on the insider's demand, and hence equilibrium quantities.  Broadly, internalization has the effect of dampening the position size, while private information may increase the insider's equilibrium allocation of risk.

We will compare when the insider (1) takes prices as given versus internalizing price impact, and (2) when the insider has a private signal versus when there is no private signal.  This leads to four cases. The first two sections deal with the effect of information asymmetry, starting with the price taking (PT) equilibrium and then moving on to the price impact (PI) equilibrium. The last two sections deal with the effect of price impact internalization, staring when there is no private information, and then considering the private signal case.  In each section we compare the equilibrium structure of demands and prices. Throughout we use the notation in \eqref{E:lambda_alpha_R_def_n}.

We start with the PT equilibrium, identifying the effect of the signal on the demands, as well as the price. Using \eqref{E:new_p_nsn_simple_n}, \eqref{E:pt_price_simple_n},  \eqref{E:pt_psiI_simple_n} and \eqref{E:no_signal_competi_simple_n} the insider's optimal demand functions satisfy
\begin{equation}\label{E:hatpsiI_decomp}
\wh{\psi}_I(g,z) - \wh{\psi}_{ns,I}(z) = \frac{(1-\consb)(p_I-p_U)}{1+\consb p_I + (1-\consb)p_U}\left(g-p_0- \frac{p_I(\consb p_I + (1-\consb)p_U) + p_U}{p_I(p_I-p_U)}\frac{z}{\alpha_I}\right).
\end{equation}
As \eqref{E:pt_PU_simple_n} implies $p_I>p_U$ we see that  (as expected) insider optimal demand relative to the no-signal case is increasing in the signal.  By contrast, insider relative demand is decreasing in the noise trader demand.

Another interesting effect concerns the outstanding supply, $\Pi = (\alpha_I + \alpha_U)\wh{\Pi} = -(\alpha_I + \alpha_U)p_0$. As both $G,Z_N$ have mean $0$, at the ex-ante level, the expected demand change is
\begin{equation}\label{E:hatpsiI_exp_diff_pt}
    \expvs{\wh{\psi}_I(G,Z_N)- \wh{\psi}_{ns,I}(Z_N)}  = \frac{(1-\consb)(p_I-p_U)}{1+\consb p_I + (1-\consb)p_U}\wh{\Pi} > 0,
\end{equation}
so that on average the insider's position increases due to the private signal when the supply is positive. Intuitively, a positive supply means agents are expected to buy the asset in equilibrium, and for the insider the private signal lowers the variance of the asset payoff (i.e. its risk) is lower. This increases her demand, implying ex-ante that she is willing to hold a larger position.

The private signal's affect on insider demand is transferred to the equilibrium price, as an increase in the insider's demand tends to increase the price as well. Indeed, from \eqref{E:pt_price_simple_n} and \eqref{E:no_signal_competi_simple_n} we find
\begin{equation*}
\begin{split}
p(h(g,z))- p_{ns}(z) &= \frac{\consb p_I + (1-\consb)p_U}{1+\consb p_I + (1-\consb)p_U}\bigg( g - p_0 + \consb\bigg(\frac{(1-\consb)(p_I-p_U)}{p_I(\consb p_I + (1-\consb)p_U)} - 1\bigg)\frac{z}{\alpha_I}\bigg).
%\frac{\consb p_I + (1-\consb)p_U}{1+\consb p_I + (1-\consb)p_U}(g-p_0)\\
%&\qquad  + \frac{\consb\left((1-\consb)(p_I-p_U)  - p_I(\consb p_I+(1-\consb)p_U)\right)}{p_I(1+\consb p_I + (1-\consb) p_U)}\frac{z}{\alpha_I},\\
\end{split}
\end{equation*}
This shows the equilibrium price is increasing in the insider signal, but is ambiguous in the noise as the factor in front of $z/\alpha_I$ can be positive or negative. At the ex-ante level
\begin{equation}\label{E:pt_price_decomp}
    \expvs{p(h(G,Z_N))-p_{ns}(Z_N)} = \frac{\consb p_I + (1-\consb)p_U}{1+\consb p_I + (1-\consb)p_U}\wh{\Pi} >0.
\end{equation}
Thus, the presence of the signal increases both the insider's expected demand and expected equilibrium price. Additionally, the equilibrium clearing condition \eqref{E:cc_psi_simple} implies  the private signal has the opposite effect on the uninformed trader's demand\footnote{The uninformed trader sees $Z_N$ (through the price) in the no-signal equilibrium, but does not see $Z_N$ in the PT equilibrium.  Therefore, we are not saying the uninformed trader sees $Z_N$ in the PT equilibrium and then adjusts her position accordingly. Rather, we are saying the effect of noise trading in the PT equilibrium is to increase the trade size of the uninformed trader over the no-signal equilibrium.}. Summarizing,

\begin{quote}
\textit{The presence of private information in the PT equilibria is expected to increase the insider's demand and the price and to decrease the uninformed trader's demand.}
\end{quote}

\subsection*{Private information effects when internalizing price impact}

The effects of private signal on equilibrium prices and demands in the PI equilibrium are similar as in the PT equilibrium. Using \eqref{E:pi_signal_simple_n}, Theorem \ref{thm:pi_simple_n} and Proposition \ref{P:no_signal_simple_n} we obtain
\begin{equation*}
\begin{split}
        \wh{\psi}_{I,\iota}(g,z) - \wh{\psi}_{ns,\iota,I}(z) &=  \frac{p_I}{1+2\wh{\yvals}}(g-p_0) + \left(\frac{\consb}{1+\consb} - \frac{\wh{\yvals}}{1+2\wh{\yvals}}\right)\frac{z}{\alpha_I}.
    \end{split}
\end{equation*}
As expected, the insider demand (relative to the no-signal case) is increasing with the private signal. Similarly, by noting the right side of \eqref{E:cubic_alt_n} is positive at $y=\consb/(1-\consb)$ one can use the arguments in the proof of Proposition \ref{prop:price_impact_signal} to show $\wh{y} > \consb/(1-\consb)$ and hence the coefficient in front of $z/\alpha_I$ is negative.  As in the PT equilibrium, at the ex-ante level the presence of signal increases the insider's demand because
\begin{equation}\label{E:hatpsiI_exp_diff_pi}
\expvs{\wh{\psi}_{I,\iota}(G,Z_N)- \wh{\psi}_{ns,\iota,I}(Z_N)}=\frac{p_I}{1+2\wh{y}}\wh{\Pi}.
\end{equation}
As for the equilibrium prices, using \eqref{E:pi_signal_simple_n}, \eqref{E:pi_price_new_simple_n} and Proposition \ref{P:no_signal_simple_n} we obtain
\begin{equation*}
    \begin{split}
        p_{\iota}(h_{\iota}(g,z)) - p_{ns,\iota}(z) &= \frac{p_I\wh{y}}{(1+p_I)(1+2\wh{y})}\left(g - p_0\right) + \left(\frac{\wh{y}(1+\wh{y})}{(1+p_I)(1+2\wh{y})} - \frac{\consb}{1-\consb^2}\right)\frac{z}{\alpha_I}.
    \end{split}
\end{equation*}
Therefore, due to the increased insider demand, the relative price change is increasing in the signal and outstanding supply. At the ex-ante level
\begin{equation*}
            \expvs{p_{\iota}(h_{\iota}(G,Z_N))-p_{ns,\iota}(Z_N)} =\frac{p_I\wh{y}}{(1+p_I)(1+2\wh{y})}\wh{\Pi},
\end{equation*}
which shows a positive difference. This is the same as in the PT case, and in fact, using \eqref{E:pt_price_decomp}, along with \eqref{E:prices_simple} and \eqref{E:mg_alt_compare} we obtain
\begin{equation}\label{E:pi_sig_px_diff}
    \begin{split}
        \expvs{p(h(G,Z_N))-p_{ns}(Z_N)}-\expvs{p_{\iota}(h_{\iota}(G,Z_N))-p_{ns,\iota}(Z_N)}= (m_g-m_{g,\iota})\wh{\Pi}.
    \end{split}
\end{equation}
Proposition \ref{prop:price_impact_signal} thus implies the expected price change caused by the presence of the private signal is lower when accounting for price impact, consistent with the reduced sensitivity with respect to public signal that internalization yields.

As in the PT equilibrium, market clearing implies the effects of the signal on the uninformed trader's demand are in thee opposite direction of the insider.  Indeed, direct calculation shows
\begin{equation*}
    \expvs{\wh{\psi}_{U,\iota}(Z_N)- \wh{\psi}_{ns,\iota,U}(Z_N)} = -\frac{p_I\consb}{(1-\consb)(1+2\wh{y})}\wh{\Pi},
\end{equation*}
which means that in contrast to the insider, the uninformed trader's demand is expected to decrease due to price impact. Thus, we may conclude
\begin{quote}
    \textit{For both the PT and PI equilibria, the presence of a private signal is expected to increase the insider's demand and price (albeit with a lower change in the PI equilibrium) and decrease the uninformed trader's demand.}
\end{quote}

\subsection*{Price impact internalization effects when there is no information asymmetry}

We now turn our attention to the effect of price impact, first assuming absence of private information. According to Proposition \ref{P:no_signal_simple_n}, and recalling that $\psi_{I,0} = \wh{\Pi}$ we see the trades satisfy
\begin{equation}\label{E:nsi_trade_diff}
    \left(\wh{\psi}_{ns,\iota,I}(z)-\psi_{I,0}\right) = -\frac{\consb}{1+\consb}\frac{z}{\alpha_I},\qquad \left(\wh{\psi}_{ns,I}(z)-\psi_{I,0}\right) = -\consb \frac{z}{\alpha_I}.
\end{equation}
Therefore,  internalization of price impact keeps the insider at the same side of the trade as in the non-internalization case, but it \textit{reduces} the magnitude of the trade. Intuitively, the insider accounts for price impact by taking a smaller position, which in turn changes the equilibrium price. Indeed, we readily get that
\begin{equation}\label{E:no_sig_px_diff}
    p_{ns,\iota}(z)-p_{ns}(z) = \frac{\consb^3}{1-\consb^2}\frac{z}{\alpha_I}.
\end{equation} 
The above implies that in the PI equilibrium, the insider has a lower demand when compared to the PT equilibrium and obtains a better price (price-impact increases the per-unit price when insider sells and decreases it when she buys). Indeed, positive $z$ means that both insider and uninformed trader sell at equilibrium and in fact positive $z$ makes the insider reveal a higher public signal\footnote{From \eqref{E:signals_together} we get that price impact changes the public signal in the direction of the noise traders' order. Indeed, as $\wh{y}>0$, we have $(h_{\iota}(g,z)-h(g,z))z\geq 0$ for all $(g,z)$. Hence, the strategically revealed signal by the insider is higher if and only if there is positive demand from the noise traders.}. This increases the demand of the uninformed trader (covers less of noise demand) and hence the interim effect is an increase in the equilibrium price. Lastly, the uninformed trader's equilibrium allocations imply the relative trade size
\begin{equation*}
    \left(\wh{\psi}_{ns,\iota,U}(z)-\psi_{U,0}\right) = -\frac{\consb}{1-\consb^2}\frac{z}{\alpha_I},\qquad \left(\wh{\psi}_{ns,U}(z)-\psi_{U,0}\right) = -\consb \frac{z}{\alpha_I}.
\end{equation*}
Therefore, when $z < 0$ both the insider and uninformed trader buy the risky asset in each equilibria. However, in the PI equilibrium, the insider reduces her position while the uninformed trader increases it. Conversely, when $z>0$ traders sell the risky asset, with the uninformed trader selling more. In particular, the uninformed increases his volume at a better price-per-unit, due to price impact. In other words, when insider buys the asset, internalization reduces her demand which in turn decreases the price and makes the uninformed trader buys more. Summarizing,
\begin{quote}
    \textit{Internalization with no signal results in a lower (resp.~higher) equilibrium position for the insider (resp.~uninformed trader) at a better price}.
\end{quote}

\subsection*{Price impact internalization effects when there is information asymmetry}

We finally consider the effect of price-impact internalization on equilibrium demands and prices in the presence of private information, which associates with our main case. To simplify the presentation and highlight the key points, we state the results at the expected value level, rather than for each realization.

We have seen that the presence of signal is expected to increase the volume of the insider's order, while the internalization is expected to decrease it (under no private signal). In other words, price impact and presence of the signal have ex-ante opposite expected effects on the insider's demand. In particular, using \eqref{E:hatpsiI_exp_diff_pt}, \eqref{E:hatpsiI_exp_diff_pi}, \eqref{E:nsi_trade_diff} and $\expvs{Z_N} = 0$ we obtain
\begin{equation}\label{E:expected_orders_alt}
\begin{split}
    \expvs{\wh{\psi}_{I,\iota}(G,Z_N)-\wh{\psi}_{I}(G,Z_N)} &= \left(\frac{p_I}{1+2\wh{y}} - \frac{(1-\consb)(p_I-p_U)}{1+\consb p_I + (1-\consb)p_U}\right)\wh{\Pi},\\
    &= p_I\left(\frac{1}{1+2\wh{y}} - \frac{(1-\consb)}{1+\consb p_I + \consa p_I(1+p_I)}\right)\wh{\Pi},
\end{split}
\end{equation}
where the second equality follows using \eqref{E:constants_nice_simple} and simplifying. The following shows the right side above may be either positive or negative.  In short, the effect of price impact prevails over the one of asymmetric information when the insider is sufficiently risk tolerant resulting in a higher expected order.

\begin{proposition}\label{P:more_trade}
The quantity in \eqref{E:expected_orders_alt} is positive as $\alpha_I \to \infty$  and negative as $\alpha_I \to 0$.  Therefore,
price impact internalization is expected to increase (resp.~decrease) the insider's position when she is sufficiently risk tolerant (resp.~risk averse). By market clearing, the opposite is expected for the uninformed trader.
\end{proposition}

\begin{remark}
By inspecting the proof of Proposition \ref{P:more_trade} we can identify other instances when $\expvs{\wh{\psi}_{I,\iota}(G,Z_N)-\wh{\psi}_{I}(G,Z_N)}$ is positive or negative.   For example, with other parameters fixed, it is positive as (i)  $p_I \to \infty$, or (ii) as $\alpha_U \to 0,\infty$ if $\alpha_I^2 p_I p_N$ is sufficiently large. Conversely, it is negative as (i) $p_I \to 0$ or (ii) as $\alpha_U \to 0,\infty$ if $\alpha_I^2 p_I p_N$ is sufficiently small.
\end{remark}

Lastly, for the equilibrium prices,  \eqref{E:pi_sig_px_diff}, \eqref{E:no_sig_px_diff} and $\expvs{Z_N} = 0$ give
\begin{equation*}
\expvs{p_{\iota}(H_{\iota})-p(H)} = (m_{g,\iota}-m_g)\wh{\Pi} .
\end{equation*}
Proposition \ref{prop:price_impact_signal} implies the right side above is negative, which means that internalization is expected to decrease (resp.~increase) the price when insider is expected to buy (resp.~sell). In other words, the expected change of the price benefits both traders (as the uninformed trader remains at the same side of trade). Connecting this fact to the expected demand changes, we may conclude that
\begin{quote}
    \textit{Due to internalization of price impact, a sufficiently low (resp.~high) risk tolerant insider is expected to buy less (resp.~more) units at a better price, while uninformed trader buys more (resp.~less)}.
\end{quote}

%------------------------------------------------%

\section{Intuition and model's predictions}\label{S:conclude}

We conclude by providing economic intuition about the predictions induced by the model.   First, we have seen that asymmetric information is ex-ante expected to increase the insider's demand for the tradeable asset. Intuitively, the private signal reduces the asset's risk (measured by variance) for the insider, which makes her willing to hold a higher allocation (positive or negative). Without strategic trading, higher insider demand is expected to increase the price. Private information tends to increase the insider's demand even under price impact. The main ex-ante expected difference is the lower increase of the price, as internalization of price impact affects the price in favor of the insider.  When the insider trades strategically, she uses her private signal to affect the equilibrium. In fact, it is exactly when the insider internalizes price impact and uninformed traders are price-takers that insider's utility is monotonically increasing with respect to the signal precision. In other words, insider's strategic trading increases the value of her private signal, making the acquisition of a better signal reasonable.

Internalization of price impact keeps the traders at the same side of trade. As we have seen, the insider hides part of her private signal, which is expected to reduce her demand and hence the equilibrium price. In other words, price impact has the opposite expected effect than the private signal. Provided that insider has long position at equilibrium, the price always decreases due to price impact, under both price impact and private signal. However, the insider's demand is lower under the presence of signal when she is also sufficiently risk averse. This is intuitive in the sense that a sufficiently risk averse trader wants to undertake less risk. If in addition to high insider's risk aversion inequality \eqref{E:less_CE_in_PI} holds, the price impact equilibrium  lowers the insider's utility. Note that \eqref{E:less_CE_in_PI} implies the uninformed traders are highly risk tolerant. This reduces the insider's allocation even further, since the uninformed traders' demand is higher and hence at market-clearing the insider's share is lower. Such lower insider demand may lead to lower ex-ante expected utility gains, as a signal of good quality (consistent with \eqref{E:less_CE_in_PI}) means the insider feels more confident to hold a higher position at equilibrium. However, under a large deviation of risk aversions, internalization of price impact prevails, demand is reduced and utility gains are lower. Note that when traders have the same risk aversion, the effect of internalization is always beneficial for the insider.

Interestingly enough, in the absence of a private signal, internalization of price impact induces higher utility gains for the insider. This is because under symmetric information, the insider does not have motive to increase her demand due to a lower asset risk, and hence the effect of price impact, i.e. buying lower quantity at a lower price, increases the expected utility. We conclude that it is the presence of asymmetric information and the deviation on risk aversions that potentially make the price-impact equilibrium disadvantageous for the insider.

We should also emphasize that as long as the insider internalizes her price impact, equilibrium prices cannot be driven to the corresponding price-taking equilibrium. This is because the uninformed trader, although he is a price-taker, realizes the insider internalizes her price impact. This makes him reduce his perceived public signal precision, and hence alters his demand function to a less elastic one. Under the Pareto initial allocation, he remains at the same side of trade with the insider and the lower equilibrium prices caused by the internalization imply higher demand for the uninformed trader, who (although price-taker) is benefited by price impact. In other words, price impact decreases the price when traders buy the asset, and at equilibrium the uninformed trader satisfies his optimal demand but at a discount. This is the reason why price impact ex-ante benefits the uninformed trader with and without asymmetric information.

\bibliographystyle{plainnat}

\bibliography{master}

\def\polhk#1{\setbox0=\hbox{#1}{\ooalign{\hidewidth
  \lower1.5ex\hbox{`}\hidewidth\crcr\unhbox0}}}
\begin{thebibliography}{33}
\providecommand{\natexlab}[1]{#1}
\providecommand{\url}[1]{\texttt{#1}}
\expandafter\ifx\csname urlstyle\endcsname\relax
  \providecommand{\doi}[1]{doi: #1}\else
  \providecommand{\doi}{doi: \begingroup \urlstyle{rm}\Url}\fi

\bibitem[Anthropelos and Kardaras(2024)]{AnthKar24}
M.~Anthropelos and C.~Kardaras.
\newblock Price impact under heterogeneous beliefs and restricted
  participation.
\newblock \emph{Journal of Economic Theory}, 215:\penalty0 105774, 2024.
\newblock ISSN 0022-0531.

\bibitem[Back(1992)]{back1992insider}
K.~Back.
\newblock Insider trading in continuous time.
\newblock \emph{Review of Financial Studies}, 5\penalty0 (3):\penalty0
  387--409, 1992.

\bibitem[Benzi and Viviani(2023)]{benzi2023solving}
M.~Benzi and M.~Viviani.
\newblock Solving cubic matrix equations arising in conservative dynamics.
\newblock \emph{Vietnam Journal of Mathematics}, 51\penalty0 (1):\penalty0
  113--126, 2023.

\bibitem[Bergemann et~al.(2021)Bergemann, Heumann, and Morris]{BerHeuMor21}
D.~Bergemann, T.~Heumann, and S.~Morris.
\newblock Information, market power, and price volatility.
\newblock \emph{The RAND Journal of Economics}, 52\penalty0 (1):\penalty0
  125--150, 2021.

\bibitem[Goldstein et~al.(2025)Goldstein, Xiong, and Yang]{GolXioYan23}
I.~Goldstein, Y.~Xiong, and L.~Yang.
\newblock Information sharing in financial markets.
\newblock \emph{Journal of Financial Economics}, 163:\penalty0 103967, 2025.

\bibitem[Gong et~al.(2022)Gong, Ke, Qiu, and Shen]{GonKeQuiShen22}
A.~Gong, S.~Ke, Y.~Qiu, and R.~Shen.
\newblock Robust pricing under strategic trading.
\newblock \emph{Journal of Economic Theory}, 199:\penalty0 105201, 2022.
\newblock Symposium Issue on Ambiguity, Robustness, and Model Uncertainty.

\bibitem[Grossman(1976)]{grossman1976efficiency}
S.~Grossman.
\newblock On the efficiency of competitive stock markets where trades have
  diverse information.
\newblock \emph{Journal of Finance}, 31\penalty0 (2):\penalty0 573--585, 1976.

\bibitem[Grossman and Stiglitz(1980)]{grossman1980impossibility}
S.J. Grossman and J.E. Stiglitz.
\newblock On the impossibility of informationally efficient markets.
\newblock \emph{American Economic Review}, 70\penalty0 (3):\penalty0 393--408,
  1980.

\bibitem[Hellwig(1980)]{hellwig1980aggregation}
M.F. Hellwig.
\newblock On the aggregation of information in competitive markets.
\newblock \emph{Journal of Economic Theory}, 22\penalty0 (3):\penalty0
  477--498, 1980.

\bibitem[Hirshleifer(1971)]{Hirsh71}
J.~Hirshleifer.
\newblock The private and social value of information and the reward to
  inventive activity.
\newblock \emph{American Economic Review}, 61\penalty0 (2):\penalty0
  561–--574, 1971.

\bibitem[Indjejikian et~al.(2014)Indjejikian, Lu, and Yang]{IndLuYan14}
R.~Indjejikian, H.~Lu, and L.~Yang.
\newblock Rational information leakage.
\newblock \emph{Management Science}, 60\penalty0 (11):\penalty0 2762--2775,
  2014.

\bibitem[Kacperczyk and Pagnotta(2019)]{KacPag19}
M.T. Kacperczyk and E.S. Pagnotta.
\newblock Chasing private information.
\newblock \emph{The Review of Financial Studies}, 32\penalty0 (12):\penalty0
  4997--5047, 03 2019.

\bibitem[Kacperczyk et~al.(2024)Kacperczyk, Nosal, and Sundaresan]{KacNosSun23}
M.T. Kacperczyk, J.~Nosal, and S.~Sundaresan.
\newblock Market power and price informativeness.
\newblock 07 2024.

\bibitem[Koijen and Yogo(2019)]{KoiYog19}
R.S.J. Koijen and M.~Yogo.
\newblock A demand system approach to asset pricing.
\newblock \emph{Journal of Political Economy}, 127\penalty0 (4):\penalty0
  1475--1515, 2019.

\bibitem[Kovalenkov and Vives(2014)]{KovViv14}
A.~Kovalenkov and V.~Vives.
\newblock Competitive rational expectations equilibria without apology.
\newblock \emph{Journal of Economic Theory}, 149:\penalty0 211--235, 2014.

\bibitem[Kyle(1985)]{kyle1985continuous}
Albert~S. Kyle.
\newblock Continuous auctions and insider trading.
\newblock \emph{Econometrica}, 53\penalty0 (6):\penalty0 1315--1335, 1985.

\bibitem[Kyle(1989)]{Kyl89}
A.S. Kyle.
\newblock Informed speculation with imperfect competition.
\newblock \emph{The Review of Economic Studies}, 56\penalty0 (3):\penalty0
  317--355, 1989.

\bibitem[Lou and Rahi(2023)]{LouRah23}
Y.~Lou and R.~Rahi.
\newblock Information, market power and welfare.
\newblock \emph{Journal of Economic Theory}, 214:\penalty0 105756, 2023.

\bibitem[Malamud and Rostek(2017)]{MalRos17}
S.~Malamud and M.~Rostek.
\newblock Decentralized exchange.
\newblock \emph{American Economic Review}, 107\penalty0 (11):\penalty0
  3320--62, 2017.

\bibitem[Morris and Shin(2002)]{MorShi02}
S.~Morris and H.~S. Shin.
\newblock Social value of public information.
\newblock \emph{American Economic Review}, 92\penalty0 (5):\penalty0
  1521--1534, December 2002.

\bibitem[Nezafat and Schroder(2023)]{NEZAFAT2023105664}
M.~Nezafat and M.~Schroder.
\newblock The negative value of private information in illiquid markets.
\newblock \emph{Journal of Economic Theory}, 210:\penalty0 105664, 2023.

\bibitem[Rochet and Vila(1994)]{rochet_vila}
J.-C. Rochet and J.-L. Vila.
\newblock Insider trading without normality.
\newblock \emph{The Review of Economic Studies}, 61\penalty0 (1):\penalty0
  131--152, 01 1994.

\bibitem[Rostek and Weretka(2012)]{RosWer12}
M.~Rostek and M.~Weretka.
\newblock Price inference in small markets.
\newblock \emph{Econometrica}, 80\penalty0 (2):\penalty0 687--711, 2012.

\bibitem[Rostek and Weretka(2015{\natexlab{a}})]{RosWer15}
M.~Rostek and M.~Weretka.
\newblock Dynamic thin markets.
\newblock \emph{Review of Financial Studies}, 28:\penalty0 2946--2992,
  2015{\natexlab{a}}.

\bibitem[Rostek and Weretka(2015{\natexlab{b}})]{RosWer15a}
M.~Rostek and M.~Weretka.
\newblock Information and strategic behavior.
\newblock \emph{Journal of Economic Theory}, 158:\penalty0 536--557,
  2015{\natexlab{b}}.
\newblock Symposium on Information, Coordination, and Market Frictions.

\bibitem[Rostek and Yoon(2024)]{RosYoo23}
M.~Rostek and J.~H. Yoon.
\newblock Equilibrium theory of financial markets: Recent developments.
\newblock Prepared for the Journal of Economic Literature, 2024.

\bibitem[Spiegel and Subrahmanyam(1992)]{spiegel1992informed}
M.~Spiegel and A.~Subrahmanyam.
\newblock Informed speculation and hedging in a noncompetitive securities
  market.
\newblock \emph{The Review of Financial Studies}, 5\penalty0 (2):\penalty0
  307--329, 1992.

\bibitem[Subrahmanyam(1991)]{subrahmanyam1991risk}
A.~Subrahmanyam.
\newblock Risk aversion, market liquidity, and price efficiency.
\newblock \emph{The Review of Financial Studies}, 4\penalty0 (3):\penalty0
  417--441, 1991.

\bibitem[Vayanos(1999)]{Vay99}
D.~Vayanos.
\newblock Strategic trading and welfare in a dynamic market.
\newblock \emph{The Review of Economic Studies}, 66\penalty0 (2):\penalty0
  219--254, 04 1999.

\bibitem[Vayanos(2001)]{Vay01}
D.~Vayanos.
\newblock Strategic trading in a dynamic noisy market.
\newblock \emph{The Journal of Finance}, 56\penalty0 (1):\penalty0 131--171,
  2001.

\bibitem[Verrecchia(1982)]{verrecchia1982}
Robert~E. Verrecchia.
\newblock Information acquisition in a noisy rational expectations economy.
\newblock \emph{Econometrica}, 50\penalty0 (6):\penalty0 1415--1430, 1982.

\bibitem[Vives(2011)]{Viv11}
X.~Vives.
\newblock Strategic supply function competition with private information.
\newblock \emph{Econometrica}, 79:\penalty0 1919--1966, 2011.

\bibitem[Vives(2014)]{Viv14}
X.~Vives.
\newblock On the possibility of informationally efficient markets.
\newblock \emph{Journal of the European Economic Association}, 12\penalty0
  (5):\penalty0 1200--1239, 2014.

\end{thebibliography}

\appendix

%---------------------------------%
%---------------------------------%

\section{The Equilibrium: General Case}\label{AS:equilibrium_general_a}

In this appendix, we extend the model of Section \ref{S:equilibrium_simple} to $d$ risky assets with general values for the payoff mean and covariance. We also precisely quantify the affine coefficients in the optimal demand functions, and derive the cubic equation \eqref{E:cubic_alt_n}. The assets have terminal payoff $X\sim N(\mu_X,P_X^{-1})$, written
\begin{equation*}
X = \mu_X + P_X^{-1/2} \E_X,
\end{equation*}
where $\E_X \sim N(0,1_d)$. The outstanding supply is $\Pi\in\RR^d$. The riskless asset price is normalized to $1$. The private signal $G$ takes the form
\begin{equation}\label{E:G_def_a}
    G = X + Z_I;\qquad Z_I  = P_I^{-1/2} \E_I,
\end{equation}
where $\E_I \sim N(0,1_d)$ is independent of $X$, so that $P_I$ is the signal precision matrix.  The noise traders have demand
\begin{equation}\label{E:ZN_def_a}
    Z_N = P_N^{-1/2} \E_N,
\end{equation}
where $\E_n\sim N(0,1_d)$ is independent of both $X$ and $\E_I$. The matrices $P_X,P_I,P_N$ lie in $\mathbb{S}^d_{++}$, the set of $d\times d$ strictly positive definite symmetric matrices, and $\mu_X\in \reals^d$. Lastly, the insider $I$ and uninformed trader $U$ are endowed with share positions $\cbra{\pi_{i,0}} = \alpha_i \whpi, i \in \cbra{I,U}$ exactly as in \eqref{E:pareto_endow}.

For a given strategy $\pi_i$, if the time $0$ price vector is $p$ then trader $i$ has risk aversion adjusted terminal wealth (see \eqref{E:tw_psi_simple})
\begin{equation*}
\mathcal{W}^{\psi_i} =  (\psi_{i,0})'p + \psi_i'(X-p) =  \whpi'p + \psi_i'(X-p);\qquad i\in\cbra{I,U}.
\end{equation*}
where the symbol $'$ denotes transposition.  The insider perceives linear price impact so that if she changes her position from $\pi_{I,0} = \alpha_{I} \wh{\Pi}$ to $\pi_I = \alpha_I \psi_I$, the price will be as in \eqref{E:impact_form_n} (for a to-be-determined vector $V_{p,\iota}$ and matrix $M_{p,\iota}$). To define the set of acceptable strategies for the insider $\A_I$, note that $G$ and $Z_N$ are the only random quantities revealed at time $0$, and hence every insider strategy must be known using the information generated by $G$ and $Z_N$.  Next, recall that our assumption (as in  \cite{rochet_vila}) is that the insider, by seeing both her private signal and the resultant price, is able to deduce the noise trader demand. But, if the insider uses a strategy $\psi$ which reveals the noise trader demand $Z_N$ through the price, it must also be that $Z_N$ is known given the information generated by $G$ and $p_{\iota}(\psi,Z_N)$. As such,  we define set of acceptable trading strategies for the insider to be
\begin{equation}\label{eq:A_I_a}
    \A_I \dfn \cbra{ \psi \in \sigma(G,Z_N) \such Z_N \in \sigma(G,p_{\iota}(\psi,Z_N))}\footnote{We will show the optimal strategy among all $\sigma(G,Z_N)$ measurable policies lies in $\A_I$, so $\A_I$ poses no restriction. However, if one does not restrict to $\A_I$ a-priori, it is not clear how to obtain the law of $X$ conditional on $\sigma(G,p_{\iota}(\psi,Z_N))$.}.
\end{equation}
Here, ``$\psi \in \sigma(G,H)$'' means that $\psi$ is $\sigma(G,H)$ measurable, and note that for any $\psi\in \A_I$, one has $\sigma(G,p_{\iota}(\psi,Z_N)) = \sigma(G,Z_N)$.  With $\A_I$ property defined, the insider's optimal investment problem is
\begin{equation}\label{E:I_problem_a}
\begin{split}
&\inf_{\psi \in \A_I} \left(\condexpvs{e^{-\left(\whpi'p_{\iota} +  \psi'\left(X-p_{\iota}\right)\right)}}{\sigma(G,Z_N)};\quad p_{\iota} = p_{\iota}\left(\psi,Z_N\right)\right).
\end{split}
\end{equation}
We next turn to the uninformed trader. To write his optimal investment problem assume there is signal $H_{\iota}$ passed to all market participants in equilibrium, and the price can be written $p_{\iota} = p_{\iota}(H_{\iota})$.  Both $H_{\iota}$ and $p_{\iota}(H_{\iota})$ are precisely defined below, but for now we take them as given. As the uninformed is a price taker, his optimization problem is
\begin{equation}\label{E:U_problem_a}
\begin{split}
&\inf_{\psi \in \sigma(H_{\iota})}\left(\condexpvs{e^{-\left(\whpi'p_{\iota} +  \psi'\left(X-p_{\iota}\right)\right)}}{\sigma(H_{\iota})};\quad p_{\iota} = p_{\iota}\left(H_{\iota}\right)\right).
\end{split}
\end{equation}

\subsection*{Transformations} To establish equilibrium, we make two convenient transformations.  First, using $\QQ_0$ from \eqref{E:Q0_def} and setting $\theta = \psi - \whpi$\footnote{$\psi\in\A_I$ if and only if $\theta\in\A_I$.} we may re-write the objective functions \eqref{E:I_problem_a}, \eqref{E:U_problem_a} as
\begin{equation}\label{E:transform_1}
    \begin{split}
        (I) & \qquad \condexpvs{e^{-\whpi'X}}{\sigma(G,Z_N)} \condexpv{\QQ_0}{}{e^{-\theta'\left(X-p_{\iota}\right)}}{\sigma(G,Z_N)};\quad p_{\iota} = p_{\iota}\left(\theta,Z_N\right),\\
        (U) & \qquad \condexpvs{e^{-\whpi'X}}{\sigma(H_{\iota})} \condexpv{\QQ_0}{}{e^{-\theta'\left(X-p_{\iota}\right)}}{\sigma(H_{\iota})};\quad p_{\iota} = p_{\iota}\left(H_{\iota}\right).
    \end{split}
\end{equation}
As such, we can work under $\QQ_0$, and using \eqref{E:pareto_endow} the clearing condition \eqref{E:cc_psi_simple} is
\begin{equation}\label{E:cc_theta_simple}
    0 = \alpha_I\wh{\theta}_I + \alpha_U\wh{\theta}_U + Z_N,
\end{equation}
which effectively makes the outstanding supply $0$.  The price impact form \eqref{E:impact_form_n} is
\begin{equation}\label{E:impact_form_a}
p_{\iota}(\theta,Z_N) = M_{p,\iota}\left(\theta + \frac{Z_N}{\alpha_I}\right) + V_{p,\iota}.
\end{equation}
Next, we de-mean the terminal payoff, price and signals. To do so, extending \eqref{E:new_p_nsn_simple_n}, define
\begin{equation}\label{E:new_p_nsn_a}
\pnsn \dfn \expv{\QQ_0}{}{X} = \mu_X - P_X^{-1}\wh{\Pi},
\end{equation}
as the equilibrium price absent private signals provided the initial endowments satisfy \eqref{E:pareto_endow}. With this notation, set
\begin{equation}\label{E:demean}
    X_0 \dfn X - \pnsn;\qquad p_{0,\iota} \dfn p_{\iota} - \pnsn;\qquad G_0 \dfn G - \pnsn;\qquad H_{\iota,0} \dfn H_{\iota} - \pnsn.
\end{equation}
Clearly, $\sigma(G,Z_N) = \sigma(G_0,Z_N)$ and $\sigma(H_{\iota,0}) = \sigma(H_{\iota})$. Also, from \eqref{E:G_def_a} we see $G_0 = X_0 + Z_I$ where $Z_I$ is a $N(0,P_I^{-1})$ random vector independent of $X_0$ under $\QQ_0$.  As such, $G_0$ under $\QQ_0$ takes the same form as $G$ under $\PP$, except now the payoff and insider signal have mean $0$. Similarly, $Z_N$ is a $N(0,P_N^{-1})$ random vector under $\QQ_0$ independent of both $X_0,Z_I$.  With this notation, from \eqref{E:transform_1} we start the the insider's problem
\begin{equation}\label{E:I_problem_aa}
    \begin{split}
        \inf_{\theta\in\A_I} \left(\condexpv{\QQ_0}{}{e^{-\theta'\left(X_0-p_{0,\iota}\right)}}{\sigma(G_0,Z_N)};\quad p_{0,\iota} = M_{p,\iota}\left(\theta+\frac{Z_N}{\alpha}\right) + V_{0,p,\iota}\right),
    \end{split}
\end{equation}
where $V_{0,p,\iota} = V_{p,\iota}-\pnsn$. We write $\pix$ as the $\PP$ precision of $X$ given $G$ (and also the $\QQ_0$ precision of $X_0$ given $G_0$), and express the matrix $M_{p,\iota}$ in terms of $\pix$ and a to-be-determined matrix $\yval$.
\begin{equation}\label{E:M_to_Y_a}
\begin{split}
\pix \dfn P_I + P_X;\qquad M_{p,\iota} = M_{p,\iota}(\yval) = \pix^{-1/2}\ \yval\ \pix^{-1/2}.
\end{split}
\end{equation}
We also write $V_{0,p,\iota} = V_{0,p,\iota}(\yval)$ to stress the dependence, and expect $\yval+\yval' \in \mathbb{S}^d_{++}$.  With this notation we obtain
\begin{lemma}\label{L:pi_opt_I_a}
Let $g,z\in\reals$. On $\cbra{G_0=g,Z_N = z}$, the unique optimizer of \eqref{E:I_problem_aa} enforces
\begin{equation}\label{E:pi_hpsi_noise_a}
\begin{split}
\wh{\theta}_{I,\iota}(g,z) + \frac{z}{\alpha_I} &= \MM\left(g + \Lambda_{\iota} z -P_I^{-1}\pix V_{0,p,\iota}(\yval)\right),
\end{split}
\end{equation}
where
\begin{equation}\label{Mbar_to_Y_a}
\begin{split}
\MM  = \MM(\yval) &\dfn \pix^{1/2}(1_d + \yval +\yval')^{-1}\pix^{-1/2}P_I,\\
\Lambda_{\iota} = \Lambda_{\iota}(\yval) &\dfn  \frac{1}{\alpha_I}P_I^{-1}\pix^{1/2}(1_d + \yval')\pix^{-1/2}.
\end{split}
\end{equation}
\end{lemma}
Motivated by \eqref{E:pi_hpsi_noise_a} and the clearing condition \eqref{E:cc_theta_simple} we define
\begin{equation}\label{E:H_def_2_a}
    \begin{split}
        H_{0,\iota} &\dfn G_0 + \Lambda_{\iota} Z_N = X_0 + Z_I + \Lambda_{\iota} Z_N,\\
        P_{U,\iota}  &= P_{U,\iota}(\yval) \dfn  \left( P_I^{-1} + \Lambda_{\iota}(\yval)P_N^{-1}\Lambda_{\iota}'(\yval)\right)^{-1}.
    \end{split}
\end{equation}
Next, by observing the price, the uninformed trader has time $0$ information $\sigma(H_{0,\iota})$ and from \eqref{E:I_problem_aa}, \eqref{E:pi_hpsi_noise_a} the de-meaned price is $p_{0,\iota} = p_{0,\iota}(H_{0,\iota}) = p_{0,\iota}(H_{0,\iota})(\yval)$ where
\begin{equation}\label{E:pi_p_funct_a}
\begin{split}
p_{0,\iota}(h_{0,\iota})(\yval) &=  M_{p,\iota}(\yval)\MM(\yval)\left(h_{0,\iota}-P_I^{-1}\pix V_{0,p,\iota}(\yval)\right)  + V_{0,p,\iota}(\yval),\\
&= M_{p,\iota}(\yval)\MM(\yval) h_{0,\iota} + \pix^{-1/2}(1_d + \yval')(1_d + \yval + \yval')^{-1}\pix^{1/2}V_{0,p,\iota}(\yval).
\end{split}
\end{equation}
With this signal and pricing function, the uninformed trader solves the problem
\begin{equation}\label{E:U_problem_aa}
    \begin{split}
        \inf_{\theta\in \sigma(H_{0,\iota})} \left(\condexpv{\QQ_0}{}{e^{-\theta'\left(X_0-p_{0,\iota}\right)}}{\sigma(H_{0,\iota})};\quad p_{0,\iota} = p_{0,\iota}(H_{0,\iota})\right),
    \end{split}
\end{equation}
and we have
\begin{lemma}\label{L:pi_opt_U_a}
Let $h_{\iota}\in\reals$. On  $\cbra{H_{0,\iota} = h_{0,\iota}}$, the unique optimizer of \eqref{E:U_problem_aa} is 
\begin{equation}\label{E:pi_hatpsiU_a}
\wh{\theta}_{U,\iota}(h_{0,\iota}) = P_{U,\iota} h_{0,\iota} - (P_{U,\iota}+P_X)p_{0,\iota}(h_{0,\iota}).
\end{equation}
\end{lemma}

In view of \eqref{E:H_def_2_a}, the market-clearing condition \eqref{E:cc_theta_simple} can be rewritten
\begin{equation}\label{E:pi_clearing_cond_psi_a}
\begin{split}
0 &= \alpha_I\left(\wh{\theta}_{I,\iota}(G_0,Z_N)  + \frac{Z_N}{\alpha_I}\right) + \alpha_{U} \wh{\theta}_{U,\iota}(H_{0,\iota}),\\
&= \alpha_I\MM(\yval)\left(H_{0,\iota} -P_I^{-1}\pix V_{0,p,\iota}(\yval)\right)\\
&\qquad\qquad  + \alpha_U\left(P_{U,\iota}(\yval) H_{0,\iota} - (P_{U,\iota}(\yval)+P_X)p_{\iota}(H_{0,\iota})(\yval)\right).
\end{split}
\end{equation}
Using \eqref{E:pi_p_funct_a} it is clear that $V_{0,p,\iota}(\yval) = 0$.  To eliminate the $H_{0,\iota}$ terms we need
\begin{equation}\label{E:pi_clear_h_aa}
    0_d = \alpha_I\MM(\yval) + \alpha_U P_{U,\iota}(\yval) - \alpha_U (P_{U,\iota}(\yval)+P_X)M_{p,\iota}(\yval)\MM(\yval).
\end{equation}
Re-writing this equation slightly, and translating back to the un-demeaned values, we obtain

\begin{proposition}\label{prop:cubic_matrix_a}
Assume $\whyval$ enforces the matrix equality
\begin{equation}\label{E:pi_clear_h_a}
\begin{split}
0_d &= \frac{\alpha_I}{\alpha_U} 1_d + P_{U,\iota}(\wh{\yval})\MM(\yval)^{-1} - (P_{U,\iota}(\wh{\yval}) + P_X)M_{p,\iota}(\wh{\yval}),
\end{split}
\end{equation}
where $M_{p,\iota}$, $\MM$, $P_{U,\iota}$  are from \eqref{E:M_to_Y_a}, \eqref{Mbar_to_Y_a}, \eqref{E:H_def_2_a}. Then, there is a price-impact equilibrium. The market signal is
\begin{equation}\label{E:H_def_3_a}
H_{\iota} = H_{0,\iota} +\pnsn =  G + \Lambda_{\iota}(\whyval)Z_N,
\end{equation}
where $\Lambda_{\iota}$ is from \eqref{Mbar_to_Y_a}. The equilibrium price is 
\begin{equation}\label{E:pi_price_new_a}
    \begin{split}
    p_{\iota}(H_{\iota}) &= \pnsn + M_{p,\iota}(\whyval)\MM(\whyval)\left(H_{\iota} - \pnsn\right).
    \end{split}
\end{equation}
The risk-aversion adjusted optimal policies are
\begin{equation}\label{E:pi_opt_policies_a}
    \begin{split}
        \wh{\psi}_{I,\iota} &= \whpi + \pix^{1/2}(1+\wh{\yval})^{-1}\pix^{-1/2}P_X\pnsn + \frac{1}{\alpha_I}\Lambda_{\iota}(\whyval)^{-1}G\\
        &\qquad - \pix^{1/2}(1+\whyval')^{-1}\pix^{1/2}p_{\iota}(H_{\iota}),\\
        \wh{\psi}_{U,\iota} &= \whpi + P_X\pnsn + P_{U,\iota} H_{\iota} - (P_{U,\iota}+P_X)p_{\iota}(H_{\iota}). 
    \end{split}
\end{equation}

\end{proposition}

In light of Proposition \ref{prop:cubic_matrix_a}, our goal is to find solutions $\wh{\yval}$ to \eqref{E:pi_clear_h_a}. This is a matrix-valued cubic equation for $\yval$ and the primary difficulty in establishing existence of solutions $\wh{\yval}$ arises due to the interaction between the precision matrices $P_I, P_N$ and $P_X$. While techniques exist for solving such equations (see \cite{benzi2023solving}), technically this would take us far beyond our intended scope, and therefore we make the following assumption, which is always valid in the case of a single asset.

\begin{assumption}\label{ass: precision_a}
$P_I = p_I P_X$ and $P_N = p_N P^{-1}_X$ for scalars $p_I,p_N > 0$.
\end{assumption}

Under this assumption we guess $\yval = \yvals 1_d$ for a scalar $\yvals > 0$.  The quantities in \eqref{E:M_to_Y_a}, \eqref{Mbar_to_Y_a} and \eqref{E:H_def_2_a} take the form
\begin{equation}\label{E:constants_nice_a}
\begin{split}
M_{p,\iota} &= \frac{y}{1+p_I}P_X^{-1}, \hspace{5pt} \MM = \frac{p_I}{1+2y} P_X, \hspace{5pt} \Lambda_{\iota} = \frac{y+1}{\alpha_I p_I}P_X^{-1}, \hspace{5pt} P_{U,\iota} = p_{U,\iota}P_X,
\end{split}
\end{equation}
where using the notation of \eqref{E:lambda_alpha_R_def_n} we have
\begin{equation}\label{E:p_uiota_def_a}
p_{U,\iota}\dfn  \left(\frac{\consa p_I^2 }{(1+y)^2 + \consa p_I}\right).
\end{equation}
By plugging in these values and simplifying one can show that \eqref{E:pi_clear_h_a} holds if and only if $\yvals$ solves \eqref{E:cubic_alt_n}. Indeed
\begin{equation*}
    \begin{split}
        &P_{U,\iota}(\wh{\yval})\MM(\yval)^{-1} + \frac{\alpha_I}{\alpha_U} 1_d- (P_{U,\iota}(\wh{\yval}) + P_X)M_{p,\iota}(\wh{\yval})\\
        &\qquad = \left(\frac{\consa p_I (1+2\yvals)}{(1+\yvals)^2 + \consa p_I} + \frac{\consb}{1-\consb} - \left(1 + \frac{\consa p_I^2}{(1+\yvals)^2 + \consa p_I}\right)\frac{\yvals}{1+p_I}\right)1_d
    \end{split}
\end{equation*}
The quantity within the parentheses simplifies to
\begin{equation*}
    \begin{split}
        \frac{\consb}{(1-\consb)((1+\yvals)^2 + \consa p_I)}\left((1+\yvals)^2\left(1- \frac{1-\consb}{\consb(1+p_I)}\yvals\right) + \frac{\consa p_I}{\consb}\left((1-\consb)\yvals + 1\right)\right),
    \end{split}
\end{equation*}
and the term within the parentheses on the right is exactly the right side of \eqref{E:cubic_alt_n}. Therefore, using Proposition \ref{prop:one_d_n} we obtain the following generalization of Theorem \ref{thm:pi_simple_n}.

\begin{theorem}\label{thm:pi_a}
Let Assumption \ref{ass: precision_a} hold. Then a price-impact equilibrium exists. Using $\wh{\yvals}$ from Proposition \ref{prop:one_d_n},  the market signal is $H_{\iota}$ from \eqref{E:H_def_3_a}.  The equilibrium price is $p_{\iota}(H_\iota)$ from \eqref{E:pi_price_new_a}.  The optimal policies for the insider and uninformed trader are from \eqref{E:pi_opt_policies_a}.
\end{theorem}

%------------------------------%

\subsection*{Price-taking equilibrium} For the price-taking results, assume there is a market signal $H$ revealed by the time $0$ price $p=p(H)$, and both traders take $p(H)$ as given. The insider has time 0 information $\sigma(H,G)$ while the uninformed trader uses $\sigma(H)$. We make the same translations as in the price-impact case (with the obvious notation) so that the insider and uninformed trader's optimal investment problems are
\begin{equation}\label{E:pt_opt_prob_a}
\begin{split}
    (I)&\qquad \inf_{\theta\in\sigma(G_0,H_0)} \condexpv{\QQ_0}{}{e^{-\theta'(X_0-p_0)}}{\sigma(G_0,H_0)},\\
    (U)&\qquad \inf_{\theta\in\sigma(H_0)} \condexpv{\QQ_0}{}{e^{-\theta'(X_0-p_0)}}{\sigma(H_0)}.
\end{split}
\end{equation}
The clearing condition is again \eqref{E:cc_theta_simple}. Provided we find a market-clearing signal $H_0$ and price $p_0$ for the above optimization problems,  the equilibrium signal and price will be $H = H_0 + \pnsn$ and $p = p_0 + \pnsn$ respectively, and the optimal risk aversion adjusted policies will be $\wh{\psi}_{i} = \wh{\theta}_i + \whpi, i\in\cbra{I,U}$. As this result does not require Assumption \ref{ass: precision_a} we state it for general parameter values in the following.

\begin{proposition}\label{P:pt_equilibrium_a}
There is a price-taking equilibrium. The market signal is
\begin{equation}\label{E:pt_signal_a}
H = H_0 + \pnsn =  G_0 + \recip{\alpha_I}P_I^{-1}Z_N + \pnsn = G + \recip{\alpha_I}P_I^{-1}Z_N.
\end{equation}
$H$ is of the same form as $G$, but  with lower precision
\begin{equation}\label{E:pt_PU_a}
P_U \dfn \left(P_I^{-1} + \frac{1}{\alpha_I^2} P_I^{-1} P_N^{-1} P_I^{-1}\right)^{-1}.
\end{equation}
The equilibrium price is 
\begin{equation}\label{E:pt_price_a}
\begin{split}
p(H) &= \pnsn + \left(\alpha_I(P_I + P_X) + \alpha_U(P_U + P_X)\right)^{-1}(\alpha_I P_I + \alpha_U P_U)\left(H-\pnsn\right).
\end{split}
\end{equation}
The optimal policies for $I$ and $U$ are
\begin{equation}\label{E:hatpsiI_a}
\begin{split}
    \wh{\psi}_I &= \whpi + P_X\pnsn + P_I G - (P_I+P_X)p(H),\\
    \wh{\psi}_U &= \whpi + P_X\pnsn + P_U H - (P_U+P_X)p(H).
\end{split}
\end{equation}
\end{proposition}

\subsection*{No private signal (NS) equilibrium}  We next state the equilibrium results, generalizing Proposition \ref{P:no_signal_simple_n}) for  no private signal. These results are valid under Assumption \ref{ass: precision_a}. Here, with the notation of \eqref{E:lambda_alpha_R_def_n} $P_U$ from \eqref{E:pt_PU_a} reduces to (compare with \eqref{E:p_uiota_def_a})
\begin{equation}\label{E:pu_def_a}
    P_U = p_U P_X;\qquad p_U\dfn \left(\frac{\consa p_I^2 }{1 + \consa p_I}\right)P_X,
\end{equation}
which shows that $p_U\to 0$ as $p_I\to 0$ on the order of $p_I^2$.

\begin{proposition}\label{P:no_signal_a}
Under Assumption \ref{ass: precision_a}, the no-signal equilibrium corresponds to $p_I = 0$. In the price-taking case, the equilibrium price and the risk aversion adjusted optimal positions are 
\begin{equation}\label{E:no_signal_competi_a}
p_{ns}=\pnsn + \consb P_X^{-1}\frac{Z_N}{\alpha_I};\qquad \wh{\psi}_{ns,I} = \wh{\psi}_{ns,U}  =  \whpi - \consb\frac{Z_N}{\alpha_I},
\end{equation}
where $\lambda$ is defined in \eqref{E:lambda_alpha_R_def_n}. In the price-impact case, the equilibrium price is
\begin{equation*}
    \begin{split}
    p_{ns,\iota} &= \pnsn  + \frac{\consb}{1-\consb^2}P_X^{-1}\frac{Z_N}{\alpha_I}.
    \end{split}
\end{equation*}
The risk aversion adjusted optimal policies for $I$ and $U$ are 
\begin{equation*}
    \begin{split}
    \wh{\psi}_{ns,\iota,I} &= \whpi - \frac{\consb}{1+\consb}\frac{Z_N}{\alpha_I};\qquad \wh{\psi}_{ns,\iota,U} = \whpi - \frac{\consb}{1-\consb^2}\frac{Z_N}{\alpha_I}.
    \end{split}
\end{equation*}

\end{proposition}

%---------------------------------------------------------%
%---------------------------------------------------------%

\section{Proofs}\label{AS:proofs}

\subsection{Proofs from Section \ref{S:equilibrium_simple}}\label{AS:proofs_of_section_1}

In this section, we first prove the cubic equation result Proposition \ref{prop:one_d_n}.  We then translate the general equilibrium results of Theorem \ref{thm:pi_a}, Proposition \ref{P:pt_equilibrium_a} and Proposition \ref{P:no_signal_a} into the forms given in Theorem \ref{thm:pi_simple_n}, Proposition \ref{P:pt_equilibrium_simple_n} and Proposition \ref{P:no_signal_simple_n} respectively.

%-------------------------------------%

\begin{proof}[Proof of Proposition \ref{prop:one_d_n}]
Define $g(y)$ as the cubic function on the right side of \eqref{E:cubic_alt_n}. It is clear that $g(0) > 0$ and $\lim_{y\to\infty} g(y) = -\infty$.  This shows there exists a solution $\wh{y}>0$ to $g(\wh{y}) = 0$.  As for uniqueness of positive solutions, straight-forward computations show for any solution to $g(y) = 0$ that
\begin{equation*}
(1+y)\dot{g}(y) = -\frac{1-\consb}{\consb}(1+y)\consa p_I - \consa p_I - (1+y)^3\frac{1-\consb}{\consb (1+p_I)}<0.
\end{equation*}
Thus, for any solution $y>-1$, $g$ strictly decreasing at $y$ and hence there is a unique solution exceeding $-1$, which is in fact positive.

\end{proof}

\subsubsection{Theorem \ref{thm:pi_simple_n} from Theorem \ref{thm:pi_a}}

Throughout, we use that in Section \ref{S:equilibrium_simple}, $P_X=1,\mu_X = 0$.  We also use Assumption \ref{ass: precision_a}. We start with the market signal. Here, \eqref{E:pi_signal_simple_n} follows from \eqref{E:H_def_3_a} and \eqref{E:constants_nice_a} (throughout at $\yvals = \wh{\yvals}$). The precision $p_{U,\iota}$ in \eqref{E:PU_pi_simple_n} is immediate from \eqref{E:lambda_alpha_R_def_n} and \eqref{E:constants_nice_a}.  The price in \eqref{E:pi_price_new_simple_n} is immediate from \eqref{E:pi_price_new_a} and \eqref{E:constants_nice_a}.

The insider and uninformed trader's policy in \eqref{E:pi_psiU_simple_n} follows from \eqref{E:pi_opt_policies_a} and \eqref{E:new_p_nsn_simple_n} using \eqref{E:constants_nice_a}.  The slope $M_{p,\iota}$ follows from \eqref{E:constants_nice_a}, and the intercept $V_{p,\iota}$ follows from \eqref{E:impact_form_a} and \eqref{E:new_p_nsn_simple_n} given the translation $V_{p,\iota} = V_{0,p,\iota} + \pnsn$ and the fact we showed (see right below \eqref{E:pi_clearing_cond_psi_a}) that $V_{0,p,\iota} = 0$. 

%---------------------------------------%

\subsubsection{Proposition \ref{P:pt_equilibrium_simple_n} from Proposition \ref{P:pt_equilibrium_a} }

Throughout, we again use that in Section \ref{S:equilibrium_simple}, $P_X=1,\mu_X = 0$, and we enforce Assumption \ref{ass: precision_a}.  The market signal in \eqref{E:pt_signal_simple_n} follows immediately from \eqref{E:pt_signal_a}. The precision $p_{U}$ in \eqref{E:pt_PU_simple_n} is immediate from \eqref{E:pu_def_a}. The price $p$ in \eqref{E:pt_price_simple_n} follows directly from \eqref{E:lambda_alpha_R_def_n},  \eqref{E:pt_price_a} and \eqref{E:pu_def_a}. The trading strategies for $I,U$ are immediate given \eqref{E:hatpsiI_a} and \eqref{E:new_p_nsn_simple_n}. 
%--------------------------------%

\subsubsection{Proposition \ref{P:no_signal_simple_n} from Proposition \ref{P:no_signal_a} }

This is an immediate consequence of Proposition \ref{P:no_signal_a} given Assumption \ref{ass: precision_a} and that in Section \ref{S:equilibrium_simple}, $P_X=1,\mu_X = 0$.

%--------------------------------------%

\subsection{Proofs from Appendix \ref{AS:equilibrium_general_a}}\label{AS:equilibrium_a}

To ease notation, throughout all proofs of this section we omit the $p,\iota$ subscript from $M_{p,\iota}, V_{0,p,\iota}$ in the price response function $p_{\iota}$ from \eqref{E:I_problem_aa}, and we omit the functional dependence on $\yval$.

\begin{proof}[Proof of Lemma \ref{L:pi_opt_I_a}]
The $\QQ_0$ law of $X_0$ given $\sigma(G_0,Z_N)$ has density
\begin{equation}\label{E:G_cond_result_a}
\frac{\qcondprobs{X_0 \in dx}{\sigma(G_0,Z_N)}}{\QQ_0\bra{X_0\in dx}} = \left(\frac{e^{-\frac{1}{2}x'P_I x + x'P_I g}}{\expv{\QQ_0}{}{e^{-\frac{1}{2}X_0'P_I X + X_0'P_I g}}}\right)\bigg|_{g=G_0}.
\end{equation}
Using this in \eqref{E:I_problem_aa}, on $\cbra{G_0=g,Z_N = z}$ the insider minimizes over $\theta=\theta(g,z)$
\begin{equation*}
\begin{split}
&\theta'\left(M\left(\theta+\frac{z}{\alpha_I}\right) + V_{0}\right) + \log\left(\frac{\expv{\QQ_0}{}{e^{ -\theta'X_0-\frac{1}{2}X_0'P_I X_0 + X_0'P_I g}}}{\expv{\QQ_0}{}{e^{-\frac{1}{2}X_0'P_I X_0 + X_0'P_I g}}}\right).\\
%&\qquad = \inf_{\psi_I}e^{(\psi_I-\psi_{I,0})'p_{\iota}\left(\chi(\psi_I, Z_N)\right)}\EE\left[\left.e^{-\psi_I'X}\right\vert \sigma(G,Z_N)\right].
\end{split}
\end{equation*}
As $X_0\stackrel{\QQ_0}{\sim} N(0,P_X^{-1})$, this specifies to (using $\theta'M\theta = (1/2)\theta'(M+M')\theta$)
\begin{equation}\label{E:pi_min_prob_a}
\begin{split}
&\frac{1}{2}\theta'\left(M+M'+\pix^{-1}\right)\theta - \theta'\pix^{-1}\left(P_Ig - \pix\left(M\frac{z}{\alpha_I} + V_{0}\right)\right).
\end{split}
\end{equation}
Using $M$ as in \eqref{E:M_to_Y_a},  the optimizer $\wh{\theta}_{I,\iota}$ is
\begin{equation}\label{E:pi_hatpsiI_a}
\begin{split}
\wh{\theta}_{I,\iota}(g,z) &= \pix^{1/2}(\yval+\yval'+1_d)^{-1}\pix^{-1/2}\bigg(P_I g  - \pix V_{0}   - \pix^{1/2}\yval\pix^{-1/2}\frac{z}{\alpha_I}\bigg).
\end{split}
\end{equation}
The identity in \eqref{E:pi_hpsi_noise_a} with $\MM,\Lambda_{\iota},\calv$ from \eqref{Mbar_to_Y_a} follow by direct computation.
\end{proof}

%---------------------%
\begin{proof}[Proof of Lemma \ref{L:pi_opt_U_a}]
Using \eqref{E:H_def_2_a} one can show the $\QQ_0$ law of $X_0$ given $\sigma(H_{0,\iota})$ has density
\begin{equation}\label{E:Hi_cond_result_a}
\frac{\qcondprobs{X_0 \in dx}{\sigma(H_{0,\iota})}}{\QQ_0\bra{X_0\in dx}}= \left(\frac{e^{-\frac{1}{2}x'P_{U,\iota} x + x'P_{U,\iota} h}}{\expv{\QQ_0}{}{e^{-\frac{1}{2}X_0'P_{U,\iota} X_0 + X_0'P_{U,\iota} h}}}\right)\bigg|_{h=H_{0,\iota}}.
\end{equation}
Using this in \eqref{E:U_problem_aa} and that $X_0\stackrel{\QQ_0}{\sim} N(0,P_X^{-1})$, on $\cbra{H_{0,\iota}=h}$ the uninformed trader minimizes over $\theta=\theta(h)$
\begin{equation}\label{E:pi_U_min_val_a}
\begin{split}
&\frac{1}{2}\theta'\left(P_{U,\iota} +P_X)^{-1}\right)\theta - \theta'(P_{U,\iota}+P_X)^{-1}\left(P_{U,\iota}h - (P_{U,\iota} + P_X)p_{\iota}(h)\right).
\end{split}
\end{equation}
\eqref{E:pi_hatpsiU_a} is immediate.
\end{proof}

%--------------------%

\begin{proof}[Proof of Proposition \ref{prop:cubic_matrix_a}]
Throughout we suppress the dependence of all functions on $\yval$. First, \eqref{E:pi_clear_h_a} follows directly from \eqref{E:pi_clear_h_aa} after dividing by $\alpha_U$ and right-multiplying by $\MM^{-1}$.  As such, provided there is a solution $\whyval$ the price impact equilibrium follows with market signal $H_{\iota}$ from \eqref{E:H_def_3_a} and price $p_{\iota}$ from \eqref{E:pi_price_new_a}.  The remainder of the proof identifies the policies in \eqref{E:pi_opt_policies_a} starting with the uninformed trader. Here, using $\wh{\theta}_{U,\iota}$ from \eqref{E:pi_hatpsiU_a}, \eqref{E:demean} and $\wh{\psi}_{U,\iota} = \whpi + \wh{\theta}_{U,\iota}$ we obtain
\begin{equation*}
\begin{split}
    \wh{\psi}_{U,\iota} &= \whpi + P_{U,\iota}\left(H_{\iota} - \pnsn\right) - \left(P_{U,\iota} + P_X\right)\left(p_{\iota}(H_{\iota}) - \pnsn\right),\\
    &= \whpi + P_X\pnsn + P_{U,\iota}H_{\iota} - \left(P_{U,\iota} + P_X\right)p_{\iota}(H_{\iota}).
\end{split}
\end{equation*}
As for $I$ our goal is to write the strategy in the same form as for $U$, i.e. as a linear function of $\whpi,\pnsn,G_0$ and $p_\iota$.  Starting with $\wh{\theta}_{I,\iota}$ in \eqref{E:pi_hatpsiI_a} (at $g=G_0 = G-\pnsn$ and using that in equilibrium $V_{0,p,\iota} = 0$) and $\wh{\psi}_{I,\iota} = \whpi + \wh{\theta}_{I,\iota}$ we find
\begin{equation*}
    \wh{\psi}_{I,\iota} = \whpi + \pix^{1/2}(\whyval+\whyval'+1_d)^{-1}\pix^{-1/2}\bigg(P_I(G-\pnsn)   - \pix^{1/2}\whyval\pix^{-1/2}\frac{Z_N}{\alpha_I}\bigg)
\end{equation*}
Next, from \eqref{E:demean}, \eqref{E:H_def_2_a} and \eqref{E:pi_p_funct_a} at $V_{0,p,\iota} = 0$ we find
\begin{equation*}
    \begin{split}
        Z_N &= \Lambda_{\iota}^{-1}\MM^{-1}M_{p,\iota}^{-1}\left(p_{\iota}(H_{\iota})-\pnsn\right) - \Lambda_{\iota}^{-1}\left(G-\pnsn\right).
    \end{split}
\end{equation*}
This gives $\wh{\psi}_{I,\iota} = \whpi + \mathbf{M}_1\pnsn + \mathbf{M}_2 G + \mathbf{M}_3 p_{\iota}(H_{\iota})$ where
\begin{equation*}
    \begin{split}
        \mathbf{M}_1 = \pix^{1/2}(\whyval+\whyval'+1_d)^{-1}\pix^{-1/2}\left(-P_I - \pix^{1/2}\whyval\pix^{-1/2}\frac{1}{\alpha_I}\Lambda_{\iota}^{-1}\left(1_d - \MM^{-1}M_{p,\iota}^{-1}\right)\right).
    \end{split}
\end{equation*}
From \eqref{E:M_to_Y_a}, \eqref{Mbar_to_Y_a} we find
\begin{equation*}
    \begin{split}
        \mathbf{M}_1 &= \pix^{1/2}(\whyval+\whyval'+1_d)^{-1}\pix^{-1/2}\bigg(-P_I - \pix^{1/2}\whyval(1_d+\whyval')^{-1}\pix^{-1/2}P_I\bigg(1_d\\
        &\qquad\qquad - P_I^{-1}\pix^{1/2}(1_d+\whyval+\whyval')\whyval^{-1}\pix^{1/2}\bigg)\bigg),\\
        %&= \pix^{1/2}(\whyval+\whyval'+1_d)^{-1}\pix^{-1/2}\bigg(-P_I - \pix^{1/2}\whyval(1_d+\whyval')^{-1}\pix^{-1/2}P_I \\
        %&\qquad\qquad + \pix^{1/2}\whyval(1_d+\whyval')^{-1}(1_d+\whyval+\whyval')\whyval^{-1}\pix^{1/2}\bigg),\\
        &= \pix^{1/2}(\whyval+\whyval'+1_d)^{-1}\pix^{-1/2}\bigg(-\pix^{1/2}(1_d+\whyval+\whyval')(1_d+\whyval')^{-1}\pix^{-1/2}P_I\\
        &\qquad\qquad + \pix^{1/2}(1_d+\whyval+\whyval')(1_d+\whyval')^{-1}\pix^{1/2}\bigg),\\
        &= \pix^{1/2}(1_d+\whyval')^{-1}\pix^{-1/2}(-P_I+\pix) = \pix^{1/2}(1_d+\whyval')^{-1}\pix^{-1/2}P_X.
    \end{split}
\end{equation*}
Next, we have
\begin{equation*}
    \begin{split}
        \mathbf{M}_2 &= \pix^{1/2}(\whyval+\whyval'+1_d)^{-1}\pix^{-1/2}\bigg(P_I + \pix^{1/2}\whyval\pix^{-1/2}\frac{1}{\alpha_I}\Lambda_{\iota}^{-1}\bigg),\\
        &= \pix^{1/2}(\whyval+\whyval'+1_d)^{-1}\pix^{-1/2}\bigg(P_I\alpha_I\Lambda_{\iota} + \pix^{1/2}\whyval\pix^{-1/2}\bigg)\frac{1}{\alpha_I}\Lambda_{\iota}^{-1},\\
        &= \pix^{1/2}(\whyval+\whyval'+1_d)^{-1}\pix^{-1/2}\bigg(\pix^{1/2}(1_d+\whyval')\pix^{-1/2} + \pix^{1/2}\whyval\pix^{-1/2}\bigg)\frac{1}{\alpha_I}\Lambda_{\iota}^{-1},\\
        &=\frac{1}{\alpha_I}\Lambda_{\iota}^{-1}.
    \end{split}
\end{equation*}
Lastly, we have
\begin{equation*}
    \begin{split}
        \mathbf{M}_3 &= -\pix^{1/2}(\whyval+\whyval'+1_d)^{-1}\pix^{-1/2}\times \pix^{1/2}\whyval\pix^{-1/2}\frac{1}{\alpha_I}\Lambda_{\iota}^{-1}\MM^{-1}M_{p,\iota}^{-1},\\
        &= -\pix^{1/2}(\whyval+\whyval'+1_d)^{-1}\whyval(1_d+\whyval')^{-1}(1_d+\whyval+\whyval')\whyval^{-1}\pix^{1/2},\\
        &= -\pix^{1/2}(1_d+\whyval')^{-1}\pix^{1/2}.
    \end{split}
\end{equation*}
\end{proof}

%---------------------------%

\begin{proof}[Proof of Proposition \ref{P:pt_equilibrium_a}]
Assume $H_0$ takes the form $H_0 = G_0 + \Lambda Z_N$ for some to-be-determined matrix $\Lambda$ and that the price is $p_0 = p_0(H_0)$. Then $\sigma(G_0,H_0) = \sigma(G_0,Z_N)$, and using the $\QQ_0$ conditional law of $X_0$ given $(G_0,Z_N)$ in \eqref{E:G_cond_result_a} as well as  $X_0 \stackrel{\QQ_0}{\sim} N(0,P_X^{-1})$,  for the insider the problem in \eqref{E:pt_opt_prob_a} is equivalent to minimizing over $\theta$
\begin{equation}\label{E:pt_min_eq_a}
\begin{split}
& \frac{1}{2}\theta'(P_I+P_X)^{-1}\theta - \theta'(P_I+P_X)^{-1}\left(P_I G_0 - (P_I+P_X)p_0(H_0)\right).
\end{split}
\end{equation}
The optimizer is $\wh{\theta}_I = P_I G_0 - (P_I+P_X)p_0(H_0)$. Similarly, as the uninformed signal is $H_0 = X_0 + (Y_I + \Lambda Z_N)$ it is of the same form as the insiders, just with precision $P_U$ instead of $P_I$.  As such, the uninformed trader's optimal policy is $\wh{\theta}_U = P_U H_0 - (P_U+P_X)p_0(H_0)$. The clearing condition \eqref{E:cc_theta_simple} is
\begin{equation*}
    0 = \alpha_I P_I \left(G_0 + \frac{1}{\alpha_I}P_I^{-1}Z_N\right) + \alpha_U P_U H_0 - (\alpha_I(P_I + P_X) + \alpha_U (P_U + P_X))p_0(H_0).
\end{equation*}
This gives $\Lambda = (1/\alpha_I)P_I^{-1}$ and hence the market signal in \eqref{E:pt_signal_a} as well as precision $p_U$ in \eqref{E:pt_PU_a}.  $p_0$ takes the form
\begin{equation*}
    p_0(H_0) = (\alpha_I(P_I + P_X) + \alpha_U (P_U + P_X))^{-1}\big(\alpha_I P_I + \alpha_U P_U)H_0,
\end{equation*}
which yields the price in \eqref{E:pt_price_a}.  Lastly, using $\wh{\psi}_I = \wh{\theta}_I +\whpi$ and $\wh{\psi}_U = \wh{\theta}_U +\whpi$ and $p(H) = p_o(H_0)-\pnsn$ we obtain \eqref{E:hatpsiI_a}.
\end{proof}
 
%------------------------------%

The proof of Proposition \ref{P:no_signal_a}, as well as those in both Sections \ref{AS:sigpx_compare} and \ref{AS:welfare} simplify if we adjust the notation in \eqref{E:lambda_alpha_R_def_n} by defining
\begin{equation}\label{E:new_lambda_alpha_R_def_a}
\cons \dfn \consa p_I = \alpha_I^2 p_I p_N; \hspace{5pt} \rtoi := \frac{\alpha_I}{\alpha_I + \alpha_U};\hspace{5pt} \rtou :=  \frac{\alpha_U}{\alpha_I+\alpha_U},\hspace{5pt} R \dfn \frac{1}{1+p_I}.
\end{equation}
The cubic equation \eqref{E:cubic_alt_n} becomes
\begin{equation}\label{E:new_cubic_a}
0 = (1+y)^2\left(1 - \frac{\rtou R}{\rtoi}y\right) + \cons \left(\frac{\rtou}{\rtoi}(1+y) + 1\right).
\end{equation}
Next, under Assumption \ref{ass: precision_a} the constants $p_{U,\iota}$ from \eqref{E:p_uiota_def_a} and $p_U$ from \eqref{E:pu_def_a} are
\begin{equation}\label{E:pi_nice_matrices_a}
    \begin{split}
        p_{U,\iota} \dfn \frac{(1-R)\cons}{R(\cons+(1+\wh{y})^2)}; \qquad p_U = \frac{(1-R)\cons}{R(1+\cons)}.
    \end{split}
\end{equation}
The matrices of Proposition \ref{P:pt_equilibrium_a} are
\begin{equation}\label{E:pt_nice_matrices_a}
    \begin{split}
        \alpha_I P_I + \alpha_U P_U &= \frac{(\alpha_I+\alpha_U)(1-R)(\rto_I + \cons)}{R(1+\cons)}P_X,\\
        \alpha_I(P_I+P_X) + \alpha_U(P_U+P_X) &= \frac{(\alpha_I+\alpha_U)(\rto_I + \cons + R\rto_U)}{R(1+\cons)}P_X.
    \end{split}
\end{equation}
Lastly, from \eqref{E:ZN_def_a}, \eqref{E:demean}, Assumption \ref{ass: precision_a} and \eqref{E:new_lambda_alpha_R_def_a} we obtain
\begin{equation}\label{E:normals_a}
    X_0 = P_X^{-1/2}\E_X;\quad G_0 = X_0 + \sqrt{\frac{R}{1-R}}P_X^{-1/2}\E_I;\qquad \frac{Z_N}{\alpha_I} = \sqrt{\frac{1-R}{\cons R}}P_X^{1/2}\E_N,
\end{equation}
where $\E_X,\E_I,\E_N$ are three independent $N(0,1_d)$ random variables under $\QQ_0$. This implies we can write $H_0$ from \eqref{E:H_def_2_a} as (see \eqref{E:pi_nice_matrices_a})
\begin{equation}\label{E:normals_h0_a}
    H_0 = X_0 + \sqrt{\frac{R}{1-R}}P_X^{-1/2}\E_I + \frac{R}{\alpha_I(1-R)}P_X^{-1}Z_N.
\end{equation}

\begin{proof}[Proof of Proposition \ref{P:no_signal_a}] We use the notation in \eqref{E:new_lambda_alpha_R_def_a}, and suppress $\wh{y}$ from $M_{p,\iota},\MM$. We start with the price taking equilibria. Using \eqref{E:pt_nice_matrices_a} and \eqref{E:normals_h0_a} the price $p(H)$ from \eqref{E:pt_price_a} is 
\begin{equation*}
p(H) = \pnsn + \frac{(1-R)(\rto_I+\cons)}{\rto_I + \cons + R\rto_U}\left(X_0 + \sqrt{\frac{R}{1-R}}P_X^{-1/2}\E_I + \frac{R}{\alpha_I(1-R)}P_X^{-1}Z_N\right).
\end{equation*}
From \eqref{E:new_lambda_alpha_R_def_a} we see that $p_I\to 0$ implies  $R\to 1,\cons \to 0$. Therefore, almost surely
\begin{equation*}
    \begin{split}
    \lim_{p_I \to 0} p(H) &= \pnsn + \frac{\rto_I}{(\rto_I+\rto_U)} P_X^{-1}\frac{Z_N}{\alpha_I} = \pnsn + \rto_I P_X^{-1}\frac{Z_N}{\alpha_I},
    \end{split}
\end{equation*}
because $\rto_U+\rto_I = 1$. This gives the pricing formula in \eqref{E:no_signal_competi_a} as $\rto=\rto_I$. As for the positions, from \eqref{E:hatpsiI_a}, \eqref{E:pi_nice_matrices_a} and \eqref{E:normals_a} we find
\begin{equation*}
    \begin{split}
        \wh{\psi}_I &= \whpi + \frac{1-R}{R}P_X\left(X_0 + \sqrt{\frac{R}{1-R}}P_X^{-1/2}\E_I\right) - \frac{1}{R}P_X\left(p(H) - \pnsn\right),\\
        \wh{\psi}_U &= \whpi + \frac{(1-R)\cons}{R(1+\cons)}P_X\left(X_0 + \sqrt{\frac{R}{1-R}}P_X^{-1/2}\E_I + \frac{R}{\alpha_I(1-R)}P_X^{-1}Z_N\right)\\
        &\qquad - \left(1 + \frac{(1-R)\cons}{R(1+\cons)}\right)\left(p(H) - \pnsn\right).
    \end{split}
\end{equation*}
Thus, almost surely as $p_I \to 0$ we have
\begin{equation*}
    \begin{split}
        \wh{\psi}_I &\to \whpi - P_X\left(\lambda_I P_X^{-1}\frac{Z_N}{\alpha_I}\right) = \whpi - \lambda_I\frac{Z_N}{\alpha_I},\\
        \wh{\psi}_U &\to \whpi - P_X\left(\lambda_I P_X^{-1}\frac{Z_N}{\alpha_I}\right) = \whpi - \lambda_I\frac{Z_N}{\alpha_I}.
    \end{split}
\end{equation*}
We next consider the price impact equilibrium.  First, from \eqref{E:H_def_2_a}, \eqref{E:constants_nice_a} and \eqref{E:normals_a} we find
\begin{equation}\label{E:pi_normals_h0_a}
    H_{0,\iota} = X_0 + \sqrt{\frac{R}{1-R}}P_X^{-1/2}\E_I + \frac{(1+\wh{y})R}{\alpha_I(1-R)}P_X^{-1}Z_N.
\end{equation}
From \eqref{E:pi_price_new_a} and \eqref{E:constants_nice_a} we obtain
\begin{equation*}
    p_{\iota}(H_{\iota}) = \pnsn + \frac{\wh{y}(1-R)}{1+2\wh{y}}\left(X_0 + \sqrt{\frac{R}{1-R}}P_X^{-1/2}\E_I + \frac{(1+\wh{y})R}{\alpha_I(1-R)}P_X^{-1}Z_N.\right).
\end{equation*}
From \eqref{E:new_lambda_alpha_R_def_a} and \eqref{E:new_cubic_a} we see that $p_I\to 0$ additionally implies $\wh{y} \to \rto_I/\rto_U$.  Therefore, almost surely, and using $\rto_U + \rto_I = 1$ we find
\begin{equation*}
    \lim_{p_I\to 0} p_{\iota}(H_{\iota})=\pnsn + \frac{\rto_I}{1-\rto^2_I}P_X^{-1}\frac{Z_N}{\alpha_I}.
\end{equation*}
As for the trading strategies, from \eqref{E:demean}, \eqref{E:pi_opt_policies_a}, \eqref{E:constants_nice_a} and \eqref{E:pi_nice_matrices_a} e find
\begin{equation*}
    \begin{split}
        \wh{\psi}_{I,\iota} &= \whpi + \frac{1}{1+\wh{y}}P_X\pnsn + \frac{1-R}{R(1+\wh{y})}P_X\left(X_0 + \sqrt{\frac{R}{1-R}}P_X^{-1/2}\E_I + \pnsn\right)  - \frac{R}{1+\wh{y}}P_X p_{\iota}(H_{\iota}),\\
        \wh{\psi}_{U,\iota} &= \whpi + P_X\pnsn + \frac{(1-R)\cons}{R(\cons+(1+\wh{y})^2} P_X\bigg(X_0 + \sqrt{\frac{R}{1-R}}P_X^{-1/2}\E_I\\
        &\qquad + \frac{(1+\wh{y})R}{\alpha_I(1-R)}P_X^{-1}Z_N + \pnsn\bigg) - \frac{\cons + R(1+\wh{y})^2}{R(\cons+(1+\wh{y})^2} P_X p_{\iota}(H_{\iota}).
    \end{split}
\end{equation*}
This gives the almost sure limits as $p_I \to 0$
\begin{equation*}
    \begin{split}
        \wh{\psi}_{I,\iota} &\to \whpi - \frac{\rto_I}{1+\rto_I}\frac{Z_N}{\alpha_I};\qquad \wh{\psi}_{U,\iota} \to \whpi - \frac{\rto_I}{1-\rto^2_I}\frac{Z_N}{\alpha_I}
    \end{split}
\end{equation*}
\end{proof}

%---------------------------%

\subsection{Proofs from Sections \ref{S:sigpx_compare_simple} \& \ref{S:demand_compare}}\label{AS:sigpx_compare}

As in Section \ref{AS:equilibrium_a}, throughout all proofs of this section we omit the $p,\iota$ subscript from $M_{p,\iota}, V_{p,\iota}$ in the price response function $p_{\iota}$ from \eqref{E:impact_form_n}.

%------------------------------------------%

\begin{proof}[Proof of Proposition \ref{prop:impact_signal_noise}]
The first statement follows directly from \eqref{E:signals_together} and $\wh{y}>0$. For the second, using \eqref{E:pi_nice_matrices_a} with $P_X = 1$ we see that $p_U>p_{U,\iota}$ is equivalent to $1 < (1+\wh{y})^2$, which holds as $\wh{y}>0$.
\end{proof}

%---------------------------%

\begin{proof}[Proof of Proposition \ref{prop:prices_simple}]
The formula for $p_{\iota}$ in \eqref{E:prices_simple} is taken directly from \eqref{E:pi_price_new_simple_n}. As for $p$, using \eqref{E:pt_price_simple_n} and \eqref{E:pt_nice_matrices_a} we find
\begin{equation}\label{E:pt_px_new_not_a}
    p(h) = \pnsn + \frac{(1-R)(\rto_I + \cons)}{\rto_I + \cons + R\rto_U}(h-\pnsn).
\end{equation}
The expression in \eqref{E:pi_price_new_simple_n} follows by the identifications in \eqref{E:new_lambda_alpha_R_def_a} and \eqref{E:lambda_alpha_R_def_n}. 
\end{proof}

%----------------------------------------%

\begin{proof}[Proof of Proposition \ref{prop:price_impact_signal}]
We retain the notation in \eqref{E:new_lambda_alpha_R_def_a} and suppress $\wh{y}$ from the $M,\MM$ functions. First, using \eqref{E:mg_alt_compare} we readily get that  $2\left(m_g-m_{g,\iota}\right)/(1-R) = 1/(1+2\wh{y}) + (\cons + \rto_I - \rto_U R)/(\cons+ \rto_I+ \rto_UR)$. This gives the result when $\cons + \rto_I> \rto_UR$.  When $\rto_U R > \cons+ \rto_I$, define $\wt{y}$ through $1/(1+2\wt{y}) = (\rto_U R - (\cons + \rto_I))/(\rto_UR + \cons + \rto_I)$. Thus, $m_g > m_{g,\iota}$ if and only if $\wh{y} < \wt{y}$.  In the proof of Proposition \ref{prop:one_d_n} we showed if we define $g$ by the right side of \eqref{E:cubic_alt_n} (see also \eqref{E:new_cubic_a}), then $g$ is strictly decreasing at $\wh{y}$.  Thus, if $g(\wt{y}) < 0$ it must be that $\wh{y} < \wt{y}$. Indeed, $\wh{y} = \wt{y}$ is not possible, and if $\wh{y} > \wt{y}$ then there must be some $\check{y} > 0$ with $g(\check{y}) = 0$, but by the uniqueness statement in Proposition \ref{prop:one_d_n} we know this is not possible as well.

It therefore suffices to show that $g(\wt{y}) < 0$. To this end,  write $p \dfn \rto_U R$ and $q \dfn \cons + \rto_I$ so that by assumption $p  > q$ and $\wt{y} = q/(p-q)$ and $\wt{y} + 1 = p/(p-q)$. In \eqref{E:new_cubic_a} we obtain
$$
g(\wt{y}) =\frac{p^2}{(p-q)^2}\left(1 - \frac{pq}{\rto_I(p-q)}\right) + \frac{\rto_Up \cons}{\rto_I(p-q)} + \cons.$$ 
As the common denominator $\rto_I(p-q)^3$ is positive, we need only show the numerator is negative.  The numerator is $\rto_Ip^2(p-q) - p^3q + \rto_U p \cons(p-q)^2 + \cons\rto_I(p-q)^3$.  If we group terms by powers of $p$ the cubic terms vanish, leaving
\begin{equation}\label{E:quad_form}
-\left(2\cons + \rto_I(1+\cons)\right)qp^2 + \left(\rto_U + 3\rto_I\right)\cons q^2 p - \cons\rto_Iq^3.
\end{equation}
Since $p>q$, the derivative of the above expression is bounded above by $-q^2\left(3\cons+ 2\rto_I\right) < 0$. Thus, \eqref{E:quad_form} is decreasing in $p$ when $p>q$ and hence bounded above by $-q^3\left(\cons+\rto_I\right) < 0$. The numerator is negative,  finishing the result.
\end{proof}
%------------------------------------------%

\begin{proof}[Proof of Proposition \ref{prop:price_impact_dd}] We again suppress $\wh{y}$. Using \eqref{E:new_lambda_alpha_R_def_a} we have
$
m_{\wh{\chi}} = (\cons+\rto_I)/\rto_U$ and $m_{\wh{\chi},\iota} = R\wh{y}.
$
Therefore, $m_{\wh{\chi}}  > m_{\wh{\chi},\iota}$ if and only if $\wh{y} < \wt{y} \dfn (\cons+\rto_I)/(\rto_U R)$. We will show $g(\wt{y}) < 0$ for $g$ defined by the right side of \eqref{E:new_cubic_a}, and this will give the result, as the proof of Proposition \ref{prop:price_impact_signal} showed for $y>0$ that $g(y) < 0$ if and only if $y > \wh{y}$.  To this end, from \eqref{E:new_cubic_a}, $g(\wt{y})=  (\rto_U R + \cons+\rtoi)^2\left(1-\cons+\rtoi/\rtoi\right)/(\rto_U^2 R^2) + \cons(\rtou R+\cons+\rtoi)/(R\rtoi) + \cons$, which equal to $-\cons\left[\rtou R(\cons +\rtoi)(1+\rtoi) + (\cons +\rtoi)^2\right]/(\rtoi\rtou^2 R^2)$, where we have used that $\rtoi+\rtou=1$. The result follows as $\cons,\rtoi, \rtou, R > 0$.
\end{proof}

%------------------------------------------%

\begin{proof}[Proof of Proposition \ref{P:more_trade}]
Using \eqref{E:new_lambda_alpha_R_def_a}, the expected value in \eqref{E:expected_orders_alt} is negative if and only if $\wh{y} > \ol{y} \dfn (\rto_I + \cons)/(2R\rto_U)$. Next, as shown in the proof of Propositions \ref{prop:price_impact_signal}, \ref{prop:price_impact_dd}, if we define function $g$ by the right side of \eqref{E:new_cubic_a} then $\wh{y} > \ol{y}$ if and only if $g(\ol{y}) > 0$. Direct calculation yields $g(\ol{y})=\cons  [(\rto_I/\cons-1)(2R\rto_U +\rto_I  + \cons)^2/(8\rto_I\rto_U^2R^2)$ +$4\cons R\rto_U^2(\rto_I + \cons  + 2R)]$. Now, let $\alpha_I \to 0$ while keeping everything else fixed.  From \eqref{E:new_lambda_alpha_R_def_a} we see $\rto_I/\cons \to \infty$, $\rto_I \to 0$, $\rto_U \to 1$ and $R$ is fixed. This implies that $g(\ol{y})>0$ and hence $\wh{y}>\ol{y}$. Conversely, when $\alpha_I \to \infty$ we see $\rto_I/\cons \to 0$, $\cons\to \infty$, $\rto_I \to 1$, $\rto_U \to 0$ and $R$ is fixed. This implies that $g(\ol{y})<0$ and hence $\wh{y}<\ol{y}$. \end{proof}

%------------------------------------------%

\subsection{Utility formulas used in Section \ref{S:welfare}}\label{AS:welfare}

We prove utility results for the setup in Appendix \ref{AS:equilibrium_general_a}, under Assumption \ref{ass: precision_a} and using the notation in \eqref{E:new_lambda_alpha_R_def_a}.  We will start with the price impact case, and then proceed to the price taking case.  Throughout we recall the initial endowments in \eqref{E:pareto_endow} and recall that $\whyval = \wh{y}1_d$ for $\wh{y}$ from Proposition \ref{prop:one_d_n}.  We also use the conditional laws (similar to \eqref{E:G_cond_result_a}, \eqref{E:Hi_cond_result_a}) for $k\in \cbra{\iota,\ }$
\begin{equation}\label{E:GH_cond_result_aa}
    \begin{split}
        \frac{\condprobs{X \in dx}{\sigma(G,Z_N)}}{\prob\bra{X\in dx}} &=  \left(\frac{e^{-\frac{1}{2}x'P_I x + x'P_I g}}{\expvs{e^{-\frac{1}{2}X'P_I X + X'P_I g}}}\right)\bigg|_{g=G},\\
        \frac{\condprobs{X \in dx}{\sigma(H_k)}}{\prob\bra{X\in dx}} &=  \left(\frac{e^{-\frac{1}{2}x'P_{U,k} x + x'P_{U,k} h}}{\expvs{e^{-\frac{1}{2}X'P_{U,k} X + X'P_{U,k} h}}}\right)\bigg|_{h=H_k}.
    \end{split}
\end{equation}
This allows us to define the ``certainty equivalents'', valid under Assumption \ref{ass: precision_a}, and using the notation in \eqref{E:new_lambda_alpha_R_def_a} 
\begin{equation}\label{E:fake_int_CES}
    \begin{split}
        \textrm{CE}^{o}_{I,0}(G) &\dfn -\alpha_I\log\left(\condexpvs{e^{-\whpi'X}}{\sigma(G,Z_N)}\right) = -\frac{\alpha_I R}{2}\left(\whpi'P_X^{-1}\whpi -2\whpi'\left(\mu_X + \frac{1-R}{R} G\right)\right),\\
        \textrm{CE}^{o}_{U,\iota}(H_{\iota}) &\dfn  -\alpha_U\log\left(\condexpvs{e^{-\whpi'X}}{\sigma(H_{\iota})}\right) = -\frac{\alpha_U}{2(1+p_{U,\iota})}\left(\whpi'P_X^{-1}\whpi - 2\whpi'(\mu_X + p_{U,\iota} H_{\iota})\right),\\
        \textrm{CE}^{o}_{U}(H) &\dfn  -\alpha_U\log\left(\condexpvs{e^{-\whpi'X}}{\sigma(H)}\right) = -\frac{\alpha_U}{2(1+p_{U})}\left(\whpi'P_X^{-1}\whpi - 2\whpi'(\mu_X + p_{U} H)\right),
    \end{split}
\end{equation}
where $p_{U,\iota},p_U$ are from \eqref{E:pi_nice_matrices_a}.  Lastly, extending \eqref{E:new_ce_nsn} to the general setup of Appendix \ref{AS:equilibrium_a}, we define the certainty equivalent absent private information 
\begin{equation}\label{E:new_ce_nsn_gen}
    \textrm{CE}^i_{nsn} \dfn \alpha_i \log\left(\expvs{e^{-\whpi'X}}\right) = \alpha_i\left(\wh{\Pi}'\mu_X - \frac{1}{2}\wh{\Pi}'P_X^{-1}\wh{\Pi}\right),\qquad i\in\cbra{I,U}.
\end{equation}

\subsection*{Price impact certainty equivalents}  
In this section our goal is to prove the following. 
\begin{proposition}\label{P:PI_CEs_a}
Using the notation of \eqref{E:new_lambda_alpha_R_def_a}, and $\wh{y}$ from Proposition \ref{prop:one_d_n}, the interim price impact certainty equivalents of Section \ref{S:welfare} are
\begin{equation}\label{E:pi_int_CES_a}
    \begin{split}
        \ceintio{I}(G,Z_N) &=  \textrm{CE}^{o}_{I,0}(G) + \frac{\alpha_I R}{2(1+2\wh{y})}\abs{\frac{1-R}{R}P_X^{1/2}(G-\pnsn) - \wh{y}P_X^{-1/2}\frac{Z_N}{\alpha_I}}^2,\\
        \ceintio{U}(H_{\iota}) &= \textrm{CE}^{o}_{U,0}(H_{\iota})  + \frac{\alpha_U\rto_I^2(1-R)^2(\cons+(1+\wh{y})^2)}{2\rto_U^2 R(\cons+R(1+\wh{y})^2)(1+2\wh{y})^2}\abs{P_X^{1/2}(H_{\iota}-\pnsn)}^2.
    \end{split}
\end{equation}
The ex-ante price impact certainty equivalents of Section \ref{S:welfare} are
\begin{equation}\label{E:PI_exCEs_a}
    \begin{split}
        \cexaio{I} &= \textrm{CE}^I_{nsn} + \frac{\alpha_I}{2}\log\left(1 + \frac{(1-R)(\cons/R+\wh{y}^2)}{(1+2\wh{y})\cons}\right),\\
        \cexaio{U} &= \textrm{CE}^U_{nsn} + \frac{\alpha_U}{2}\log\left(1+ \frac{\rto_I^2(1-R)(\cons+(1+\wh{y})^2}{\rto_U^2\cons R(1+2\wh{y})^2}\right).
    \end{split}
\end{equation}
\end{proposition}

\begin{proof}[Proof of Proposition \ref{P:PI_CEs_a}]
We start with the insider. From \eqref{E:I_problem_a}, \eqref{E:transform_1}, \eqref{E:demean} and \eqref{E:I_problem_aa}
\begin{equation*}
    \begin{split}
        e^{-\frac{1}{\alpha_I}\ceintio{I}} &= \condexpvs{e^{-\whpi'p_{\iota} - \wh{\psi}_{I,\iota}'(X - p_{\iota})}}{\sigma(G,Z_N)},\\
        &= \condexpvs{e^{-\whpi'X}}{\sigma(G,Z_N)} \condexpv{\QQ_0}{}{e^{\wh{\theta}_{I,\iota}'(X_0-p_{0,\iota})}}{\sigma(G_0,Z_N)}.
    \end{split}
\end{equation*}
From \eqref{E:M_to_Y_a}, \eqref{E:pi_min_prob_a} (recalling that in equilibrium, $V_0$ therein is $0$), \eqref{E:constants_nice_a} and \eqref{E:new_lambda_alpha_R_def_a}
\begin{equation*}
    \condexpv{\QQ_0}{}{e^{\wh{\theta}_{I,\iota}'(X_0-p_{0,\iota})}}{\sigma(G_0,Z_N)} = e^{-\frac{R}{2(1+2\wh{y})}\abs{\frac{1-R}{R} P_X^{1/2}G_0- \wh{y}P_X^{-1/2}\frac{Z_N}{\alpha_I}}^2}.
\end{equation*}
Therefore, using \eqref{E:fake_int_CES} we obtain the interim certainty equivalent for $I$ in \eqref{E:pi_int_CES_a}.  For the ex-ante certainty equivalent we must compute
\begin{equation*}
    e^{-\frac{1}{\alpha_I}\cexaio{I}} = \expvs{e^{-\whpi'p_{\iota} - \wh{\psi}_{I,\iota}'(X - p_{\iota})}} = \expvs{e^{-\whpi'X}}\expv{\QQ_0}{}{\condexpv{\QQ_0}{}{e^{\wh{\theta}_{I,\iota}'(X_0-p_{0,\iota})}}{\sigma(G_0,Z_N)}}.
\end{equation*}
Using \eqref{E:new_ce_nsn_gen} and the above we obtain
\begin{equation*}
   e^{-\frac{1}{\alpha_I}\cexaio{I}} = e^{-\frac{1}{\alpha_I}\textrm{CE}^I_{nsn}} \expv{\QQ_0}{}{e^{-\frac{R}{2(1+2\wh{y})}\abs{\frac{1-R}{R} P_X^{1/2}G_0 - \wh{y}P_X^{-1/2}\frac{Z_N}{\alpha_I}}^2}}.
\end{equation*}
If we re-write this using the notation of \eqref{E:new_lambda_alpha_R_def_a}, \eqref{E:normals_a} we obtain
\begin{equation*}
    e^{-\frac{1}{\alpha_I}\cexaio{I}} = e^{-\frac{1}{\alpha_I}\textrm{CE}^I_{nsn}} \expv{\QQ_0}{}{e^{-\frac{R}{2(1+2\wh{y})}\abs{\frac{1-R}{R}\left(\E_X + \sqrt{\frac{R}{1-R}}\E_I\right) - \wh{y}\sqrt{\frac{1-R}{\cons R}}\E_N}^2}}.
\end{equation*}
We can write
\begin{equation*}
    \frac{1-R}{R}\left(\E_X + \sqrt{\frac{R}{1-R}}\E_I\right) - \wh{y}\sqrt{\frac{1-R}{\cons R}}\E_N = \sqrt{\frac{(1-R)(\cons + \wh{y}^2 R)}{\cons R^2}}\E,
\end{equation*}
where $\E\sim N(0,1_d)$ under $\QQ_0$ $N(0,1_d)$. This implies
\begin{equation*}
    e^{-\frac{1}{\alpha_I}\cexaio{I}} =  e^{-\frac{1}{\alpha_I}\textrm{CE}^I_{nsn}} \left(1 + \frac{(1-R)(\cons/R+\wh{y}^2)}{(1+2\wh{y})\cons}\right)^{-1/2},
\end{equation*}
which gives the formula in \eqref{E:PI_exCEs_a}. We now consider the uninformed trader. From \eqref{E:U_problem_a}, \eqref{E:transform_1}, \eqref{E:demean} and \eqref{E:U_problem_aa} we obtain
\begin{equation*}
    e^{-\frac{1}{\alpha_U} \ceintio{U}} = \condexpvs{e^{-\whpi'p_{\iota} - \wh{\psi}_{U,\iota}'(X - p_{\iota})}}{\sigma(H_{\iota})} = \condexpvs{e^{-\whpi'X}}{\sigma(H_{\iota})} \condexpv{\QQ_0}{}{e^{\wh{\theta}_{U,\iota}'(X_0-p_{0,\iota})}}{\sigma(H_{0,\iota})}.
\end{equation*}
From \eqref{E:pi_U_min_val_a} we find
\begin{equation*}
    \condexpv{\QQ_0}{}{e^{\wh{\theta}_{U,\iota}'(X_0-p_{0,\iota})}}{\sigma(H_{0,\iota})} = e^{-\frac{1}{2}\abs{(P_{U,\iota}+P_X)^{-1/2}\left(P_{U,\iota}H_{0,\iota} - (P_{U,\iota}+P_X)p_{0,\iota}(H_{0,\iota})\right)}^2}.
\end{equation*}
From \eqref{E:pi_price_new_a}, \eqref{E:constants_nice_a} and \eqref{E:pi_nice_matrices_a} calculation shows
\begin{equation*}
    \begin{split}
        &(P_{U,\iota}+P_X)^{-1/2}\left(P_{U,\iota}H_{0,\iota} - (P_{U,\iota}+P_X)p_{0,\iota}(H_{0,\iota})\right)\\
        &\qquad = \sqrt{\frac{(1-R)^2(1+\wh{y})^2(\cons-R\wh{y}(1+\wh{y}))^2}{(\cons+R(1+\wh{y})^2)R(\cons+(1+\wh{y})^2)(1+2\wh{y})^2}} P_X^{1/2}H_{0,\iota}.
    \end{split}
\end{equation*}
From \eqref{E:new_cubic_a} we deduce
\begin{equation*}
    (1+\wh{y})^2(\cons - R\wh{y}(1+\wh{y}))^2 = \left(\frac{\rto_I}{\rto_U}(\cons + (1+\wh{y})^2)\right)^2,
\end{equation*}
which gives
\begin{equation*}
    \condexpv{\QQ_0}{}{e^{\wh{\theta}_{U,\iota}'(X_0-p_{0,\iota})}}{\sigma(H_{0,\iota})} = e^{-\frac{1}{2}\frac{\rto_I^2(1-R)^2(\cons+(1+\wh{y})^2)}{\rto_U^2 R(\cons+R(1+\wh{y})^2)(1+2\wh{y})^2}\abs{P_X^{1/2} H_{0,\iota}}^2}.
\end{equation*}
The interim certainty equivalent for $U$ in \eqref{E:pi_int_CES_a} follows from \eqref{E:fake_int_CES} and \eqref{E:demean}. For the ex-ante certainty equivalent we must compute
\begin{equation*}
    \begin{split}
        e^{-\frac{1}{\alpha_U}\cexaio{U}} = \expvs{e^{-\whpi'p_{\iota} - \wh{\psi}_{U,\iota}'(X - p_{\iota})}} &= \expvs{e^{-\whpi'X}}\expv{\QQ_0}{}{\condexpv{\QQ_0}{}{e^{\wh{\theta}_{U,\iota}'(X_0-p_{0,\iota})}}{\sigma(H_{0,\iota})}},\\
        &= e^{-\frac{1}{\alpha_U}\textrm{CE}^U_{nsn}}\expv{\QQ_0}{}{e^{-\frac{\rto_I^2(1-R)^2(\cons+(1+\wh{y})^2)}{2\rto_U^2 R(\cons+R(1+\wh{y})^2)(1+2\wh{y})^2} \abs{P_X^{1/2} H_{0,\iota}}^2}}.
    \end{split}
\end{equation*}
Using \eqref{E:normals_a} and \eqref{E:pi_normals_h0_a} we deduce
\begin{equation*}
    P_X^{1/2}H_{0,\iota} =  \sqrt{\frac{\cons + R(1+\wh{y})^2}{\cons(1-R)}}\E,
\end{equation*}
where $\E\sim N(0,1_d)$ under $\QQ_0$.  This gives
\begin{equation*}
    \begin{split}
        \expvs{e^{-\whpi'p_{\iota} - \wh{\psi}_{U,\iota}'(X - p_{\iota})}} &= e^{-\frac{1}{\alpha_U}\textrm{CE}^U_{nsn}}\left(1+ \frac{\rto_I^2(1-R)(\cons+(1+\wh{y})^2}{\rto_U^2\cons R(1+2\wh{y})^2}\right)^{-1/2},
    \end{split}
\end{equation*}
and hence the ex-ante certainty equivalent for $U$ in \eqref{E:PI_exCEs_a}.

\end{proof}

%------------------------------------------------%
\subsection*{Price-taking certainty equivalents} Here, we consider the price-taking case, with the goal of proving the following. 

\begin{proposition}\label{P:PT_CEs_a}
Using the notation of \eqref{E:new_lambda_alpha_R_def_a}, the interim price taking certainty equivalents of Section \ref{S:welfare} are
\begin{equation}\label{E:pt_int_CES_a}
    \begin{split}
        \ceinto{I}(G,Z_N) &=  \textrm{CE}^{o}_{I,0}(G)\\
        &\qquad + \frac{\alpha_I R}{2(\rto_I + \cons + R\rto_U)^2}\abs{(1-R)\rto_U P_X^{1/2}(G-\pnsn)- (\rto_I+\cons)P_X^{-1/2}\frac{Z_N}{\alpha_I}}^2,\\
        \ceinto{U}(H) &= \textrm{CE}^{o}_{U,0}(H_{\iota})  + \frac{\alpha_U\rto_I^2(1-R)^2R(1+\cons)}{2(\rto_I+\cons + R\rto_U)^2(R+\cons)}\abs{P_X^{1/2}(H-\pnsn)}^2.
    \end{split}
\end{equation}
The ex-ante price taking certainty equivalents of Section \ref{S:welfare} are
\begin{equation}\label{E:pt_exCEs_a}
    \begin{split}
        \cexao{I} &= \textrm{CE}^I_{nsn} + \frac{\alpha_I}{2}\log\left(1 + \frac{(1-R)(\rto_U^2\cons R + (\rto_I + \cons)^2)}{\cons(\rto_I + \cons + R\rto_U)^2}\right),\\
        \cexao{U} &= \textrm{CE}^U_{nsn} + \frac{\alpha_U}{2}\left(1+ \frac{\rto_I^2(1-R)R(1+\cons)}{\cons(\rto_I+\cons+R\rto_U)^2}\right).
    \end{split}
\end{equation}
\end{proposition}

\begin{proof}[Proof of Proposition \ref{P:PT_CEs_a}]
We start with the insider. From \eqref{E:transform_1}, \eqref{E:demean} and \eqref{E:pt_opt_prob_a}
\begin{equation*}
    \begin{split}
        e^{-\frac{1}{\alpha_I}\ceinto{I}} &= \condexpvs{e^{-\whpi'X}}{\sigma(G,Z_N)} \condexpv{\QQ_0}{}{e^{\wh{\theta}_{I}'(X_0-p_{0})}}{\sigma(G_0,Z_N)}.
    \end{split}
\end{equation*}
From \eqref{E:pt_signal_a}, \eqref{E:pt_min_eq_a}, \eqref{E:new_lambda_alpha_R_def_a} and \eqref{E:pt_px_new_not_a}  we obtain after some simplification
\begin{equation*}
    \condexpv{\QQ_0}{}{e^{\wh{\theta}_{I}'(X_0-p_{0})}}{\sigma(G_0,Z_N)} = e^{-\frac{R}{(\rto_I + \cons + R\rto_U)^2}\abs{(1-R)\rto_U P_X^{1/2}G_0- (\rto_I+\cons)P_X^{-1/2}\frac{Z_N}{\alpha_I}}^2}.
\end{equation*}
Therefore, using \eqref{E:fake_int_CES} we obtain the interim certainty equivalent for $I$ in \eqref{E:pt_int_CES_a}.  For the ex-ante certainty equivalent we must compute
\begin{equation*}
    e^{-\frac{1}{\alpha_I}\cexao{I}} = \expvs{e^{-\whpi'p - \wh{\psi}_{I}'(X - p)}} = \expvs{e^{-\whpi'X}}\expv{\QQ_0}{}{\condexpv{\QQ_0}{}{e^{\wh{\theta}_{I}'(X_0-p_{0})}}{\sigma(G_0,Z_N)}}.
\end{equation*}
Using \eqref{E:new_ce_nsn_gen} and the above we obtain
\begin{equation*}
   e^{-\frac{1}{\alpha_I}\cexao{I}} = e^{-\frac{1}{\alpha_I}\textrm{CE}^I_{nsn}} \expv{\QQ_0}{}{e^{-\frac{R}{(\rto_I + \cons + R\rto_U)^2}\abs{(1-R)\rto_U P_X^{1/2}G_0- (\rto_I+\cons)P_X^{-1/2}\frac{Z_N}{\alpha_I}}^2}}.
\end{equation*}
If we re-write this using the notation of \eqref{E:new_lambda_alpha_R_def_a}, \eqref{E:normals_a} we obtain
\begin{equation*}
    e^{-\frac{1}{\alpha_I}\cexao{I}} = e^{-\frac{1}{\alpha_I}\textrm{CE}^I_{nsn}} \expv{\QQ_0}{}{e^{-\frac{R}{(\rto_I + \cons + R\rto_U)^2}\abs{(1-R)\rto_U\left(\E_X + \sqrt{\frac{R}{1-R}}\E_I\right) - (\rto_I+\cons)\sqrt{\frac{1-R}{\cons R}}\E_N}^2}}.
\end{equation*}
We can write
\begin{equation*}
    (1-R)\rto_U\left(\E_X + \sqrt{\frac{R}{1-R}}\E_I\right) - (\rto_I+\cons)\sqrt{\frac{1-R}{\cons R}}\E_N = \sqrt{\frac{(1-R)(\rto_U^2\cons R + (\rto_I+\cons)^2)}{\cons R}}\E,
\end{equation*}
where $\E\sim N(0,1_d)$ under $\QQ_0$ $N(0,1_d)$. This implies
\begin{equation*}
    e^{-\frac{1}{\alpha_I}\cexao{I}} =  e^{-\frac{1}{\alpha_I}\textrm{CE}^I_{nsn}} \left(1 + \frac{(1-R)(\rto_U^2\cons R + (\rto_I + \cons)^2)}{\cons(\rto_I + \cons + R\rto_U)^2}\right)^{-1/2},
\end{equation*}
which gives the formula in \eqref{E:pt_exCEs_a}. We now consider the uninformed trader. From \eqref{E:transform_1}, \eqref{E:demean} and \eqref{E:pt_opt_prob_a}
\begin{equation*}
    e^{-\frac{1}{\alpha_U} \ceinto{U}} = \condexpvs{e^{-\whpi'X}}{\sigma(H)} \condexpv{\QQ_0}{}{e^{\wh{\theta}_{U}'(X_0-p_{0})}}{\sigma(H_{0})}.
\end{equation*}
From \eqref{E:pt_min_eq_a} (with $P_U$ replacing $P_I$ and $H_0$ replacing $G_0$), \eqref{E:new_lambda_alpha_R_def_a}, \eqref{E:pi_nice_matrices_a} and \eqref{E:pt_px_new_not_a}  we obtain after some simplification
\begin{equation*}
    \condexpv{\QQ_0}{}{e^{\wh{\theta}_{U}'(X_0-p_{0})}}{\sigma(H_0)} = e^{-\frac{\rto_I^2(1-R)^2 R(1+\cons)}{2(\rto_I+\cons + R\rto_U)^2(R+\cons)}\abs{P_X^{1/2} H_0}^2}.
\end{equation*}
Therefore, using \eqref{E:fake_int_CES} we obtain the interim certainty equivalent in \eqref{E:pt_int_CES_a}. As for the ex-ante certainty equivalent, we must compute
\begin{equation*}
    e^{-\frac{1}{\alpha_U}\cexao{U}} = \expvs{e^{-\whpi'p - \wh{\psi}_{U}'(X - p)}} = \expvs{e^{-\whpi'X}}\expv{\QQ_0}{}{\condexpv{\QQ_0}{}{e^{\wh{\theta}_{U}'(X_0-p_{0})}}{\sigma(H_0)}}.
\end{equation*}
Using \eqref{E:new_ce_nsn_gen} and the above this is
\begin{equation*}
   e^{-\frac{1}{\alpha_U}\cexao{U}} = e^{-\frac{1}{\alpha_U}\textrm{CE}^I_{nsn}} \expv{\QQ_0}{}{e^{-\frac{\rto_I^2(1-R)^2R(1+\cons)}{2(\rto_I+\cons + R\rto_U)^2(R+\cons)}\abs{P_X^{1/2} H_0}^2}}.
\end{equation*}
From \eqref{E:normals_a}, \eqref{E:normals_h0_a} we can write
\begin{equation*}
    H_0 = \sqrt{\frac{\cons + R}{\cons(1-R)}}P_X^{-1}\E,
\end{equation*}
where $\E\sim N(0,1_d)$ under $\QQ_0$.  This gives
\begin{equation*}
    \begin{split}
        e^{-\frac{1}{\alpha_U}\cexao{U}} = e^{-\frac{1}{\alpha_U}\textrm{CE}^U_{nsn}}\left(1+ \frac{\rto_I^2(1-R)R(1+\cons)}{\cons(\rto_I+\cons+R\rto_U)^2}\right)^{-1/2},
    \end{split}
\end{equation*}
and hence the ex-ante certainty equivalent for $U$ in \eqref{E:pt_exCEs_a}.
\end{proof}

\begin{remark}\label{R:nice_way}
If we move from the notation of \eqref{E:new_lambda_alpha_R_def_a} to the notation in \eqref{E:lambda_alpha_R_def_n} we obtain the ex-ante certainty equivalents
\begin{equation*}
\begin{split}
\cexaio{I} &=  \textrm{CE}_{nsn}^{I}  + \frac{\alpha_I}{2}\log\left(1 + \frac{\consa p_I(1+p_I) + \wh{y}^2}{\consa (1+p_I)(1+2\wh{y})}\right),\\
\cexaio{U} &=  \textrm{CE}_{nsn}^{U}  + \frac{\alpha_U}{2} \log\left(1 + \frac{\consb^2(\consa p_I+ (1+\wh{y})^2)}{(1-\consb)^2\consa (1+2\wh{y})^2} \right),\\
\cexao{I} &= \textrm{CE}_{nsn}^{I}  + \frac{\alpha_I}{2}\log\left(1 + \frac{(1-\consb)^2\consa p + (1+p_I)(\consb+\consa p_I)^2}{\consa(1 + \consb p_I + \consa p_I (1+p_I))^2}\right),\\
\cexao{U} &= \textrm{CE}_{nsn}^{U}  + \frac{\alpha_U}{2}\log\left(1 +\frac{\consb^2(1+\consa p_I)}{\consa(1 + \consb p_I + \consa p_I(1+p_I))^2}\right).
\end{split}
\end{equation*}
\end{remark}

%----------------------------------%

\subsection{Proofs from Section \ref{S:welfare}}\label{AS:welfare_2}

\begin{proof}[Proof of Theorem \ref{T:pimono}]
Set $c = (1-\consb)/\consb$, $p = p_I$, $\wh{y}(p_I) = y(p)$ so \eqref{E:cubic_alt_n} is
\begin{equation}\label{E:fpimon3}
    0 = (1+y(p))^2\left(1-\frac{cy(p)}{1+p}\right) + \consa p c(1+y(p)) + \consa p.
\end{equation}
Next define the function (c.f. \eqref{E:fpimon}) as $   Q(p,y) = [\consa p(1+p) + y^2]/[\consa (1+p)(1+2y)]$ so that
%\begin{equation*}
%    Q(p,y) = \frac{\consa p(1+p) + y^2}{\consa (1+p)(1+2y)}
%\end{equation*}
\begin{equation*}
    \partial_p Q(p,y) = \frac{\consa(1+p)^2 -y^2}{\consa (1+p)^2(1+2y)};\quad \partial_y Q(p,y) = \frac{2(y(1+y)-\consa p(1+p))}{\consa (1+p)(1+2y)^2}.
\end{equation*}
Additionally, from \eqref{E:fpimon3} we deduce
\begin{equation*}
\begin{split}
0 &= \left(\consa +\consa c (1+y(p)) + \frac{cy(p)(1+y(p))^2}{(1+p)^2}\right)\\
&\qquad - \left(\frac{c(1+y(p))^2}{1+p} + 2(1+y(p))\left(\frac{cy(p)}{1+p}-1\right) - \consa c p\right) \partial_p y(p).
\end{split}
\end{equation*}
Using \eqref{E:fpimon3} one can show the quantity in front of $\partial_p y(p)$ is strictly positive so that $y(p)$ is increasing in $p$ and hence
\begin{equation*}
    \begin{split}
        \partial_p y(p) & = \frac{\consa +\consa c (1+y(p)) + \frac{cy(p)(1+y(p))^2}{(1+p)^2}}{\frac{c(1+y(p))^2}{1+p} + 2(1+y(p))\left(\frac{cy(p)}{1+p}-1\right) - \consa c p}.
    \end{split}
\end{equation*}
By the chain rule
\begin{equation}\label{E:partial_f}
    \partial_p \phi_{\iota}(p) = \partial_p Q(p,y(p)) + \partial_y Q(p,y(p))\partial_p y(p),
\end{equation}
and we wish to show the right side above is positive for all $p>0, c>0,\consa > 0$.  To do this we will change perspective. Namely, from \eqref{E:fpimon3} we see that $y(p)\geq (1+p)/c$ and in fact $y(p) = (1+p)/c$ when $\consa = 0$.  Thus, let us substitute $y(p) = (1+p)/c + z(p)$ so that $z(p)$ solves (uniquely over the positive reals) the equation
\begin{equation*}
    0 = - \frac{c z}{1+p}\left(1 + z + \frac{1+p}{c}\right)^2 + \consa p c\left(1+z + \frac{1+p}{c}\right) + \consa p.
\end{equation*}
Now, fix $p>0,c>0$ and think about $z = z(\consa)$.  It is straight-forward to show that $z$ is strictly increasing in $\consa$ with extreme values $z(0) = 0$ and $z(\infty) = \infty$.  Therefore, there is no loss in generality in fixing $p>0,c>0,z>0$ and setting
\begin{equation*}
    \consa = \frac{z(1+p+c + cz)^2}{pc(1+p)(2+p+c+cz)}.
\end{equation*}
Plugging in $y = (1+p)/c +z$ and $\consa$ as above we obtain
\begin{equation*}
    \begin{split}
    \partial_p Q(p,y(p)) &= \frac{zc(1+p)(1+p+c+cz)^2 - p(1+p+cz)(2+p+c+cz)}{z(1+p)(2+2p+c+2z)(1+p+c+cz)^2},\\
    \partial_y Q(p,y(p)) &= \frac{2pc((1+p)(1+p+c+cz) + (1+p+cz))}{z(1+p+c+cz)(2+2p+c+2z)^2},\\
    \partial_p y(p) &= \frac{(1+p+c+cz)(2+p+c+cz)(zc(1+p)+p(1+p+cz))}{pc(1+p)(cz + (1+p+c+2cz)(2+p+c+cz))}.
    \end{split}
\end{equation*}
At this point, if one plugs these values into the right side of \eqref{E:partial_f} and takes a common denominator, the numerator is a sixth order polynomial in $z$. Furthermore, one can directly verify each of the coefficients in the polynomial is positive for all $p>0,c>0$, giving the result.
\end{proof}

%----------------------------------%

\begin{proof}[Proof of Proposition \ref{prop: exante_U_pipt}]
From \eqref{E:PI_better_IU} we see that $\cexaio{U} \geq \cexao{U}$ is equivalent to
\begin{equation*}
    f(y) \dfn \frac{\consa p_I + (1+\wh{y})^2}{(1+2\wh{y})^2} \geq \frac{(1-\consb)^2 (1+\consa  p_I)}{(1+p_I\consb + \consa p_I(1+p_I))^2} \rdfn \ell.
\end{equation*}
The map $y\to f(y)$ is strictly decreasing with $f(0) = \consa p_I+1$ and $f(\infty) = 1/4$.  Furthermore, as $\ell$ is evidently decreasing in $\consb \in (0,1)$ we know
\begin{equation*}
    0 < \ell < \frac{1+\consa p_I}{(1+\consa p_I(1+p_I))^2} < f(0),
\end{equation*}
and hence $k\leq 1/4$ implies $\cexaio{U} \geq \cexao{U}$.  For $1/4 < k$, the positive root of $\consa p_I + (1+y)^2 = \ell(1+2y)^2$ is
\begin{equation*}
    \check{y} = \frac{1-2\ell+\sqrt{(4\ell-1)\consa p_I+\ell}}{4\ell-1}.
\end{equation*}
As shown in the proof of Proposition \ref{prop:one_d_n}, if we define $g(y)$ by the cubic function in \eqref{E:cubic_alt_n}, then $g(y) > 0$ for $0 < y <\wh{y}$ and $g(y) < 0$ for $y>\wh{y}$.  Therefore, if $g(\check{y}) < 0$ then $\check{y} > \wh{y}$ and
\begin{equation*}
    \frac{\consa p_I + (1+\wh{y})^2}{(1+2\wh{y})^2} >\frac{\consa p_I + (1+\check{y})^2}{(1+2\check{y})^2} =\ell,
\end{equation*}
giving the result. It therefore suffices perform the following check: (1) Fix $0 < \consb < 1$, $\consa, p >0$ and let $\ell = \ell(\consb,\consa,p_I)$ as above; (2) If $\ell \leq 1/4$ then $\cexaio{U} \geq \cexao{U}$; (3) If $\ell > 1/4$ then set $\check{y} = \check{y}(\consb,\consa,p_I)$ and $g = g(\check{y})$ as above.  If $g(\check{y}) < 0$ then $\cexaio{U} \geq \cexao{U}$. This check can easily be performed by any software tool and one always obtains that either $\ell \leq 1/4$ or $g(\check{y}) < 0$.
\end{proof}

%------------------------------%
\begin{proof}[Proof of Proposition \ref{P:gamma_I_asympt_n}]

We prove this result in the setup in Appendix \ref{AS:equilibrium_general_a}. Here, recall  the certainty equivalent absent private information \eqref{E:new_ce_nsn_gen}. Using \eqref{E:pareto_endow} and \eqref{E:no_signal_competi_a} we obtain
\begin{equation}\label{E:ce_nsn_limits}
    \begin{split}
        \lim_{\alpha_I \to 0} \frac{\textrm{CE}^{I}_{nsn}}{\alpha_I}  &= \frac{1}{\alpha_U}\Pi'\mu_X - \frac{1}{2\alpha_U^2}\Pi'P_X^{-1}\Pi.
    \end{split}
\end{equation}
We first consider when $\alpha_I = \alpha_U$ so that (see \eqref{E:lambda_alpha_R_def_n}) $\consb=1/2$. As such, \eqref{E:PI_better_IU} implies $\cexaio{I} \geq \cexao{I}$ is equivalent to
\begin{equation*}
    \begin{split}
        f(\wh{y})\dfn \frac{\consa p_I(1+p_I)+ \wh{y}^2}{(1+2\wh{y})} &\geq \frac{(1+p_I)\left(\consa p_I + (1+p_I)(1+2\consa p_I)^2\right)}{(2 + p_I+ 2\consa p_I(1+p_I))^2} \rdfn k,
    \end{split}
\end{equation*}
where $\wh{y}$ solves \eqref{E:cubic_alt_n} with $\lambda = 1/2$.  But $\dot{f}(y) = 2(y(1+y)-\consa p_I(1+p_I))/(1+2y)^2$, hence $f$ is minimized $(0,\infty)$ at  $y_0 = (1/2)(\sqrt{1+4\consa p_I(1+p_I)} - 1)$, and this value enforces $f(y_0)= y_0$. Therefore, $f(\wh{y}) \geq f(y_0) = y_0$ and hence if
\begin{equation}\label{E:check_1}
y_0  \geq k \, \Leftrightarrow \, \frac{1}{2}\left(\sqrt{1+4\consa p_I(1+p_I)} - 1\right)\geq \frac{(1+p_I)\left(\consa p_I + (1+p_I)(1+2\consa p_I)^2\right)}{(2 + p_I+ 2\consa p_I(1+p_I))^2},
\end{equation}
then $\cexaio{I} \geq \cexao{I}$. Otherwise, $f(y_0) < k$ is strictly less than the right side above, and denote by $\check{y}$ the unique $y>y_0$ such that $f(\check{y}) = k$. As \eqref{E:cubic_alt_n} implies
\begin{equation*}
\wh{y}(1+\wh{y}) - \consa p_I(1+p_I) = \frac{(1+p_I)\left((1+\wh{y})^2 + \consa p_I\right)}{1+\wh{y}} > 0,
\end{equation*}
$\dot{f}(\wh{y}) > 0$ and hence $\wh{y} > y_0$.  Therefore, $\cexaio{I} \geq \cexao{I}$ is equivalent to $\wh{y} \geq \check{y}$, which as shown in the proof of Proposition \ref{prop:one_d_n}, is equivalent to $g(\check{y}) \geq 0$ for $g$ defined by the cubic function in \eqref{E:cubic_alt_n}.  However, it can easily be checked by any software tool that either \eqref{E:check_1} holds, or $g(\check{y}) \geq 0$. This gives the result for $\alpha_I = \alpha_U$.

We next fix $\alpha_U$ and consider $\alpha_I \to 0$. This corresponds to both $\consb \to 0$ and $\consa \to 0$. As such, we again will write $\consa = \alpha_I^2 p_N$ and appeal to \eqref{E:cubic_alt_n} when analyzing $\wh{y}$.  For $\cexao{I}$ one can see
\begin{equation}\label{E:111}
    \frac{(1-\consb)^2\consa p_I + (1+p_I)(\consb+\consa p_I)^2}{\consa(1 + \consb p_I + \consa p_I (1+p_I))^2} \to \frac{\alpha_U^2 p_I p_N + 1 + p_I}{\alpha_U^2 p_N}.
\end{equation}
On the other hand, \eqref{E:cubic_alt_n} implies $\wh{y}/\alpha_I \to (1+p_I)/\alpha_U$. Therefore,
\begin{equation*}
    \frac{\consa p_I(1+p_I) + \wh{y}^2}{\consa (1+p_I)(1+2\wh{y})} \to \frac{\alpha_U^2 p_I p_N + 1 + p_I}{\alpha_U^2 p_N}.
\end{equation*}
The above limits give
\[
\lim_{\alpha_I \to 0} \frac{\cexao{I}}{\alpha_I} = \lim_{\alpha_I \to 0} \frac{\cexaio{I}}{\alpha_I} = \frac{1}{\alpha_U}\Pi'\mu_X - \frac{1}{\alpha_U^2}\Pi'P_X^{-1}\Pi + \frac{d}{2}\log\left(1 + \frac{\alpha_U^2 p_I p_N + 1 + p_I}{\alpha_U^2 p_N}\right).
\]
Now, the third order approximation for $\wh{y}$ around $\alpha_I = 0$ is
\begin{equation*}
    \wt{y} \dfn \frac{\alpha_I}{\alpha_U}(1+p_I)\left(1+\alpha_I p_I p_N(\alpha_U-\alpha_I p_I)\right)
\end{equation*}
and calculation shows $\limsup_{\alpha_I\to 0} \alpha_I^{-4} |\wh{y}-\wt{y}| < \infty$. As such, in the limit
\begin{equation*}
    \lim_{\alpha_I\to 0} \frac{1}{\alpha_I^2}\left(\frac{\consa p_I(1+p_I) + \wh{y}^2}{\consa (1+p_I)(1+2\wh{y})}  - \frac{(1-\consb)^2\consa p_I + (1+p_I)(\consb+\consa p_I)^2}{\consa(1 + \consb p_I + \consa p_I (1+p_I))^2}\right),
\end{equation*}
we can substitute $\wt{y}$ in for $\wh{y}$ to obtain $    (1+p_I)^2\left(1+p_I-\alpha_U^2p_I p_N\right)/\alpha_U^4 p_N$. 
Using the above, limit \eqref{E:111}, and the fact that $\lim_{x\to 0} \log[1+C(x)x]/x = C$ if $C(x) \to C$, we get
\begin{equation*}
        \lim_{\alpha_I\to 0} \frac{\cexaio{I}-\cexao{I}}{\alpha_I^3}  = \frac{d}{2}\frac{(1+p_I)^2(1+p_I-\alpha_U^2 p_I p_N)}{\alpha_U^2(\alpha_U^2p_I p_N + 1+p_I)}.
    \end{equation*}
The latter completes the proof of the limiting arguments when $\alpha_I\rightarrow 0$.

\end{proof}

%--------------------------------------%

We finish by proving Proposition \ref{prop: nosigint_simple} in the setup of Appendix \ref{AS:equilibrium_general_a}, and recalling $\textrm{CE}^i_{nsn}$ from \eqref{E:new_ce_nsn_gen}.

\begin{proposition}\label{prop: nosigint}
Let Assumption \ref{ass: precision_a} hold. As $p_I \to 0$ we obtain almost surely
\begin{equation*}
\begin{split}
\lim_{p_I\to 0} \ceinto{I}(G,Z_N) &= \textrm{CE}^I_{nsn}   + \frac{\consb^2}{2\alpha_I}Z_N'P_X^{-1}Z_N,\\
\lim_{p_I\to 0}  \ceintio{I}(G,Z_N) &= \textrm{CE}^I_{nsn} +\frac{\consb^2}{2\alpha_I(1-\consb^2)}Z_N'P_X^{-1}Z_N.
\end{split}
\end{equation*}
For the uninformed trader we obtain almost surely
\begin{equation*}
\begin{split}
\lim_{p_I\to 0} \ceinto{U}(H) &= \textrm{CE}^U_{nsn} + \frac{\alpha_U \consb^2}{\alpha_I^2}Z_N'P_X^{-1}Z_N,\\
\lim_{p_I\to 0} \ceintio{U}(H_{\iota}) &=  \textrm{CE}^U_{nsn} + \frac{\alpha_U\consb^2}{2\alpha_I^2(1-\consb)^2(1+\consb)^2} Z_N'P_X^{-1}Z_N.
\end{split}
\end{equation*}

\end{proposition}

\begin{proof}[Proof of Proposition \ref{prop: nosigint}]
From \eqref{E:normals_a}, \eqref{E:fake_int_CES}, \eqref{E:pi_int_CES_a}, \eqref{E:pt_int_CES_a} we obtain
\begin{equation*}
    \begin{split}
        \ceinto{I}(G,Z_N) &=  \textrm{CE}^{o}_{I,0}(G)\\
        &\qquad + \frac{\alpha_I R}{2(\rto_I + \cons + R\rto_U)^2}\abs{\rto_U (1-R)P_X^{1/2}(G-\pnsn)- (\rto_I+\cons)P_X^{-1/2}\frac{Z_N}{\alpha_I}}^2,\\
        \ceintio{I}(G,Z_N) &=  \textrm{CE}^{o}_{I,0}(G) + \frac{\alpha_I R}{2(1+2\wh{y})}\abs{\frac{1-R}{R}P_X^{1/2}(G-\pnsn) - \wh{y}P_X^{-1/2}\frac{Z_N}{\alpha_I}}^2,\\
    \end{split}
\end{equation*}
where
\begin{equation*}
    \begin{split}
        \textrm{CE}^{o}_{I,0}(G) &= -\frac{\alpha_I R}{2}\left(\whpi'P_X^{-1}\whpi -2\whpi'\left(\mu_X + \frac{1-R}{R} G\right)\right),\\
        (1-R) G &= (1-R)X_0 + \sqrt{R(1-R)}P_X^{-1/2}\E_I.
    \end{split}
\end{equation*}
From \eqref{E:cubic_alt_n} and \eqref{E:new_lambda_alpha_R_def_a} we see that $p_I\to 0$ implies $\wh{y} \to \rto_I/\rto_U$, $\cons\to 0$ and $R \to 1$. The results for $\ceinto{I}(G_,Z_N)$ and $\ceintio{I}(G,Z_N)$ follow from \eqref{E:new_ce_nsn_gen} as $\rto = \rto_I$ and $1 = \rto_I + \rto_U$.  For the uninformed trader we similarly obtain (also using \eqref{E:normals_h0_a}, \eqref{E:pi_normals_h0_a})
\begin{equation*}
    \begin{split}
        \ceinto{U}(H_{\iota}) &= \textrm{CE}^{o}_{U,0}(H_{\iota})  + \frac{\alpha_U\rto_I^2R(1+\cons)}{2(\rto_I+\cons + R\rto_U)^2(R+\cons)}\times \abs{(1-R)P^{1/2}_X(H-\pnsn)}^2,\\
        \ceintio{U}(H_{\iota}) &= \textrm{CE}^{o}_{U,0}(H_{\iota})\\
        &\qquad + \frac{\alpha_U\rto_I^2(\cons+(1+\wh{y})^2)}{2\rto_U^2 R(\cons+R(1+\wh{y})^2)(1+2\wh{y})^2}\times \abs{(1-R)P^{1/2}_X(H_{\iota}-\pnsn)}^2,
    \end{split}
\end{equation*}
where
\begin{equation*}
    \begin{split}
        (1-R)H_0 &= (1-R)X_0 + \sqrt{R(1-R)}P_X^{-1/2}\E_I + \frac{R}{\alpha_I}P_X^{-1}Z_N,\\
        (1-R)H_{0,\iota} &= (1-R)X_0 + \sqrt{R(1-R)}P_X^{-1/2}\E_I + \frac{(1+\wh{y})R}{\alpha_I}P_X^{-1}Z_N.
    \end{split}
\end{equation*}
The results for $U$ readily follow.

\end{proof}

\end{document}